\documentclass[10pt]{article}
\usepackage{amsfonts,amsmath}
\usepackage{mathbbol,graphics,epsf}
\usepackage{verbatim,amssymb,amscd,eucal,bbm,setspace,numprint}
\usepackage{gag}
\usepackage{enumerate}

\usepackage{graphicx}
\usepackage{vaucanson-g}
\usepackage{array,multirow,colortbl}%
\usepackage{rotating}
\usepackage{a4wide}

\newcommand{\trans}{\mathfrak{T}}
\newcommand{\inver}{\mathfrak{i}}
\def\eref#1{(\ref{#1})}
\newcommand{\pres}[1]{\langle{#1}\rangle}
\newcommand{\presm}[1]{\pres{{#1}}_{+}}
\def\cH{{\mathcal H}}
\def\m{\mu}
\def\s{\sigma}
\def\FR{\mathbf{FR}}
\def\SK{\mathbf{automgrp}}
\def\GAP{\mathbf{GAP}}
\def\resp{\hbox{\textit{resp.}} }

\def\Z{\mathbb{Z}}
\def\N{\mathbb{N}}

\newcommand{\StackFourLabels}[4]{%
   \renewcommand{\arraystretch}{0.75}%
   \begin{array}{c}#1\\ #2 \\ #3 \\ #4 \end{array}%
   \renewcommand{\arraystretch}{1.333}}

\newcommand{\StackEightLabels}[8]{%
   \renewcommand{\arraystretch}{0.75}%
   \begin{array}{c}#1\\ #2 \\ #3 \\ #4 \\ #5 \\ #6 \\ #7 \\ #8 \end{array}%
   \renewcommand{\arraystretch}{1.333}}

\def\mz{{\mathfrak m}}
\def\dz{{\mathfrak d}}
\def\iz{{\mathfrak i}}
\newcommand{\Ac}{\mathcal{A}}
\newcommand{\Bc}{\mathcal{B}}

\newcommand{\id}[1][]{\text{id}_{#1}}

\newcommand{\G}{\grEng{{\cal A}}}

\newcolumntype{E}{!{\vrule width 1.3pt}}
\def\neg{\!\;\!}
\def\Call{\mathbf{W}}
\def\Ci{\mathbf{I}}
\def\Cir{\mathbf{I\neg R}}
\def\Cbir{\mathbf{B\neg I\neg R}}
\def\Cjir{\mathbf{J\neg I\neg R}}
\def\Cijir{\mathbf{\iz J\neg I\neg R}}
\def\Cdijir{\mathbf{\dz\iz J\neg I\neg R}}
\def\Cji{\mathbf{J\neg I}}
\def\Cdji{\mathbf{\dz J\neg I}}
\def\Cnot{\mathbf{N}}
\def\us{\char`\_}
\def\col{6mm}
\def\colcol{8mm}
\def\colcolcol{10mm}

\newtheorem{theorem}{Theorem}[section]
\newtheorem{proposition}[theorem]{Proposition}
\newtheorem{definition}[theorem]{Definition}
\newtheorem{corollary}[theorem]{Corollary}
\newtheorem{lemma}[theorem]{Lemma}

\newenvironment{proof}{{\bf Proof.}}%
{\hspace*{\fill}\(\square\)}

\begin{document}

\title{On the Finiteness Problem for Automaton (Semi)groups}

 \author{Ali Akhavi\thanks{GREYC - CNRS UMR 6072 \& Universit\'e de
     Caen, France} \and Ines Klimann\thanks{LIAFA - CNRS UMR
     7089 \& Universit\'e Paris Diderot-Paris~7, France}
   \and Sylvain Lombardy\thanks{LIGM - CNRS UMR 8049 \& Universit\'e
     Paris Est, France} \and Jean Mairesse\footnotemark[2] \and
   Matthieu Picantin\footnotemark[2]}

\date{\today}

\maketitle

\begin{abstract}
This paper addresses a decision problem highlighted by~Grigorchuk,
Nekrashevich, and~Sushchanski{\u\i}, namely the finiteness problem for
automaton (semi)groups. 
For semigroups, we give an effective sufficient but not necessary condition for
finiteness and, for groups, an effective necessary but not sufficient
condition. The efficiency of the new criteria is demonstrated
by testing all Mealy automata with small stateset and alphabet.
Finally, for groups, we provide a necessary and sufficient
condition that does not directly lead to a decision procedure.  
\end{abstract}

\section{Introduction}\label{s:intro}

\emph{Automaton (semi)groups} --- short for semigroups generated by
Mealy automata or groups generated by invertible Mealy automata ---
were formally introduced a half century ago (for details, see~\cite{clas32} and references therein). 
Two decades later, important results have started revealing their full potential.
In particular, contributing to the Burnside problem, \cite{aleshin,grigorchuk1} construct
Mealy automata generating particularly simple infinite torsion groups,
and, answering to the Milnor problem, \cite{brs,grigorchukMilnor}
produce Mealy automata generating the first examples of (semi)groups with intermediate growth. 

Since these pioneering works, a substantial theory continues to develop using various methods, ranging from finite
automata theory to geometric group theory, and various viewpoints from
self-similarity to natural actions on regular rooted trees
(see~\cite{bgn,bgs,bs,clas32,gns,gsu,nek} for groups
and~\cite{brs,cain,gns,mal,min,sst} for semigroups) 
and never ceases to show that automaton (semi)groups possess multiple interesting and
sometimes unusual features. 

The classical decision problems have been investigated for automaton groups and semigroups: the
word problem is solvable \cite{cain,gns} while the conjugacy
problem has recently been proved to be unsolvable \cite{conjugacy}.
Here we address the \emph{finiteness problem}, that is, the question of 
the existence of an algorithm that takes as input a Mealy
automaton and decides if the generated (semi)group is
finite (see~\cite[Problem~7.2.1(b)]{gns}). 
Since the word problem is solvable, then
a semidecision procedure for the finiteness problem simply consists 
of enumerating all the elements.

Three results related to the finiteness problem have to be mentioned here. First, the finiteness problem is solved for the special class of semigroups generated by (dual) Cayley machines~(see~\cite{cain,mal,min,sst}) by using semigroup theory and especially the Green's relations machinery.
Second, the class of those automata which always generate finite (semi)groups independently of their output function has been completely characterized (see~\cite{anto,antoberk,russ}).
Third, the class of so-called ``bounded'' (invertible) automata where all the states have growth degree at most~0 has been thoroughly studied
and the solution to the order problem (see~\cite{sidkiconjugacy,sidki}) yields an infiniteness criterion.
Observe that these three classes correspond to very special structures for the concerned Mealy automata.

Two $\GAP$ packages are dedicated to
automaton (semi)groups: $\FR$ by~Bartholdi and~$\SK$ by Muntyan
and~Savchuk \cite{FR,GAP4,sav}. Both include specific (in)finiteness
tests. 
Besides the three results above-mentioned, all that was known up to now
about the finiteness 
question for automaton groups happened to be somehow summarized in the
documentation for~$\FR$:

\begin{quote}\small
The order of [an automaton] group is computed as follows: if all [the
states have growth degree at most~0], then enumeration will succeed
in computing the order. If the action of the group is primitive, and
it comes from a bireversible automaton, then the Thompson-Wielandt
theorem is tested against [$\ldots$]
see~\cite[Prop.~2.1.1]{BM}. Then, $\FR$ attempts to find whether
the group is level-transitive (in which case it would be
infinite). Finally, it attempts to enumerate the group's elements,
testing at the same time whether these elements have infinite order.

\nobreak\noindent\medskip
Needless to say, none except the first few steps are guaranteed to succeed.
\end{quote}

\smallskip\noindent In this paper, we give several new criteria for testing (in)finiteness,
that could easily be added to the $\FR$ and $\SK$ packages. 
The original ingredients in these packages mainly come from geometric group
theory. Our new notions and tools --- like \emph{helix
  graphs} and \emph{minimization-dualization} --- are automata-theoretic in
nature and most often work in the general setting of semigroups.
The common idea is to put a special emphasis on the \emph{dual automaton},
obtained by exchanging the roles of stateset and alphabet.
 The stepping stone is Proposition~\ref{pr:duale-finitude} stating that any
 Mealy automaton generates a finite semigroup if and only
 if so does its dual. 
The general strategies vary by analyzing a Mealy automaton and
its dual either alternatively --- see the minimization-dualization
reduction in Section~\ref{s:reduction} --- or both together as a
whole --- see the helix graph construction in~Section~\ref{s:helix}.

In Section~\ref{s:reduction}, we give an effective sufficient but not
necessary condition for finiteness using minimization-dualization. 
Focusing on those invertible automata with invertible dual, and using helix
graphs, Section~\ref{s:helix} provides an effective necessary but not sufficient
condition for finiteness, and also a non-effective necessary and
sufficient condition. The decidability of the finiteness problem
remains open. 

\smallskip

Gathering the new criteria with
the previously known ones allows to decide the (semi)group (in)finiteness
for substantially more Mealy automata. In Table~\ref{tbl-intro}, we report
on the 
results of the experimentation carried out on: $(i)$ all 3-letter 2-state Mealy automata; $(ii)$ all 3-letter 3-state invertible or reversible 
Mealy automata. The first three columns are the number of automata treated
successfully respectively
by previously known criteria, our new criteria, and the union of both. The last column is
the total number of Mealy automata. The automata are counted up to isomorphism. 

\begin{table}[ht]
\centering
\caption{Some results of the experimentations to decide (in)finiteness with old and new criteria.\label{tbl-intro}}
{\begin{tabular}{lE>{\centering }m{24mm}|>{\centering }m{18mm}|>{\centering }m{20mm}E>{\centering }m{12mm}E}
\cline{2-5}
 & {\bf\scriptsize previous criteria} &{\bf\scriptsize new criteria}&{\bf\scriptsize previous+new}&{\bf\scriptsize~~~total~~~}\tabularnewline
\hline
\multicolumn{1}{ElE}{\,general (3,2)} 		& 398 			& \numprint{1130} 		& \numprint{1214}		& \numprint{4003} \tabularnewline\hline
\multicolumn{1}{ElE}{\,inv. or rev. (3,3)\,} 	& \numprint{78721} 	& \numprint{100924} 	& \numprint{172737} 	& \numprint{236558} \tabularnewline \hline
\end{tabular}}
\end{table}

\noindent More detailed experimental results are given in
Section~\ref{sec-experimentations} 
and a gallery of meaningful examples is given in Table~\ref{tbl-examples}.

\section{Preliminaries}
\label{s:prelim}

Let $S$ be a finite and non-empty set. We denote by 
$\trans_S$ the set of functions from~$S$ to~$S$, and we denote by
$\perm_S$ the set of bijections from~$S$ to~$S$. 

\subsection{Mealy automaton}
If one forgets about initial and final states, a {\em
(finite, deterministic and complete) automaton} $\aut{A}$ is a
triple 
\(
\bigl( A,\Sigma,\delta = (\delta_i: A\rightarrow A )_{i\in \Sigma} \bigr)
\),
where the \emph{set of states}~$A$
and the \emph{alphabet}~$\Sigma$ are non-empty finite sets, and
where the $\delta_i$'s are functions.
In a condensed way, the automaton is identified with $\delta$, that is
an element of~$\trans_A^{\Sigma}$.

\smallskip

A \emph{Mealy automaton} is a quadruple 
\[
\bigl( A, \Sigma, \delta = (\delta_i: A\rightarrow A )_{i\in \Sigma},
\rho = (\rho_x: \Sigma\rightarrow \Sigma  )_{x\in A} \bigr) \:,
\]
such that both $(A,\Sigma,\delta)$ and $(\Sigma,A,\rho)$ are
automata. 
Another standard terminology for Mealy automaton would
be: letter-to-letter transducer with the same input and output
alphabets. 
A Mealy automaton is identified with an element of
$\trans_A^{\Sigma}\times \trans_{\Sigma}^A$.

Graphically, a Mealy automaton is represented by a labelled directed
graph with:
\[
\mathrm{nodes}: \ A, \qquad \text{arcs (transitions)}: \ x
\stackrel{i | j}{\longrightarrow} y \ \mathrm{if}
\ \delta_i(x)=y\text{ and }\rho_x(i)=j \:.
\]
The notation $x\stackrel{\mot{u}| \mot{v}}{\longrightarrow} y$ with
$\mot{u}=u_1\cdots u_n$, $\mot{v}=v_1\cdots v_n$ is a shorthand for
the existence of a path $x \stackrel{u_1|v_1}{\longrightarrow} x_1
\stackrel{u_2|v_2}{\longrightarrow} x_2 \cdots x_{n-1}
\stackrel{u_{n}|v_n}{\longrightarrow} y$ in $\aut{A}$. 

\smallskip

Two examples of Mealy automata are given in Fig.~\ref{fi-Maut}. 

\begin{figure}[ht]
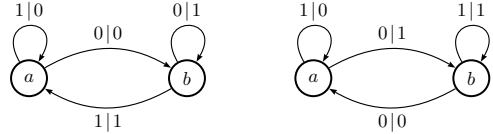

\begin{center}
\TinyPicture
\VCDraw{%
\begin{VCPicture}{(-6,-1)(10,2)}
\LargeState
\State[a]{(-5,0)}{A}
\State[b]{(0,0)}{B}
\State[a]{(4,0)}{C}
\State[b]{(9,0)}{D}
\LArcL[0.5]{A}{B}{\IOL{0}{0}}
\LArcL[0.5]{B}{A}{\IOL{1}{1}}
\LoopN[0.5]{A}{\IOL{1}{0}}
\LoopN[0.5]{B}{\IOL{0}{1}}
\LArcL[0.5]{C}{D}{\IOL{0}{1}}
\LArcL[0.5]{D}{C}{\IOL{0}{0}}
\LoopN[0.5]{C}{\IOL{1}{0}}
\LoopN[0.5]{D}{\IOL{1}{1}}
\end{VCPicture}
}
\end{center}
\caption{Two Mealy automata.}
\label{fi-Maut} 
\end{figure}

In a Mealy automaton $(A,\Sigma, \delta, \rho)$, the sets $A$ and
$\Sigma$ play dual roles. So we may consider the \emph{dual (Mealy)
  automaton} defined by
\(
\dz(\aut{A}) = (\Sigma,A, \rho, \delta)
\).
Alternatively, we can define the dual Mealy automaton via the set of
its transitions:
\begin{equation}
x \stackrel{i\mid j}{\longrightarrow} y \ \in \aut{A} \quad \iff \quad i \stackrel{x\mid y}{\longrightarrow} j \ \in
\dz(\aut{A}) \:. 
\label{eq-dual}
\end{equation}
In what follows, it is often pertinent to consider a Mealy automaton
and its dual together, that is to work with the pair $\{\aut{A},
\dz(\aut{A})\}$. A pair of dual Mealy automata is represented in
Fig.~\ref{fi-Mautdual}. 

\begin{figure}[ht]
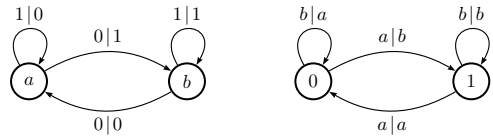

\begin{center}
\TinyPicture
\VCDraw{%
\begin{VCPicture}{(-6,-.5)(10,2)}
\LargeState
\State[a]{(-5,0)}{A}
\State[b]{(0,0)}{B}
\State[0]{(4,0)}{C}
\State[1]{(9,0)}{D}
\LArcL[0.5]{A}{B}{\IOL{0}{1}}
\LArcL[0.5]{B}{A}{\IOL{0}{0}}
\LoopN[0.5]{A}{\IOL{1}{0}}
\LoopN[0.5]{B}{\IOL{1}{1}}
\LArcL[0.5]{C}{D}{\IOL{a}{b}}
\LArcL[0.5]{D}{C}{\IOL{a}{a}}
\LoopN[0.5]{C}{\IOL{b}{a}}
\LoopN[0.5]{D}{\IOL{b}{b}}
\end{VCPicture}
}
\end{center}
\caption{A pair of dual Mealy automata.}
\label{fi-Mautdual} 
\end{figure}

Consider a Mealy automaton $\aut{A}\in \trans_A^{\Sigma}\times
\perm_{\Sigma}^A$.
Let $A^{-1}=\{x^{-1}, x \in A\}$ be a disjoint copy
of $A$. The \emph{inverse (Mealy) automaton} $\inverse{\aut{A}}\in \trans_{A^{-1}}^{\Sigma}\times
\perm_{\Sigma}^{A^{-1}}$ is 
defined by the set of its transitions:
\begin{equation}
x \stackrel{i\mid j}{\longrightarrow} y \ \in \aut{A} \quad \iff \quad x^{-1}
\stackrel{j\mid i}{\longrightarrow} y^{-1} \ \in \aut{A}^{-1} \:. 
\label{eq-inv}
\end{equation}

\smallskip

Let us call respectively \emph{dualization} (denoted $\dz$) and {\em
  inversion} (denoted $\inver$) the two
transformations on transitions defined in \eref{eq-dual} and
\eref{eq-inv}. Starting with a transition and alternating the
dualization and inversion transformations, we obtain eight
transitions. (In the process, we also define $\Sigma^{-1}=\{x^{-1}, x
\in \Sigma\}$, a disjoint copy of $\Sigma$; and we set $(A^{-1})^{-1}=A$ and
$(\Sigma^{-1})^{-1}=\Sigma$.)

\smallskip

Now consider a Mealy automaton $\aut{A}$ identified with its set of
transitions, and apply the same transformations to $\aut{A}$. We
obtain eight sets of transitions that we denote by: 
\[
\aut{A}, \ \dz(\aut{A}), \ \inver(\aut{A}), \ \dz\inver(\aut{A}),
\ \inver\dz(\aut{A}), \ \dz\inver\dz(\aut{A}),
\  \inver\dz\inver(\aut{A}), \ \dz\inver\dz\inver(\aut{A}) =
\inver\dz\inver\dz(\aut{A}) \:.
\]
If $\aut{A}\in \trans_A^{\Sigma}\times
\perm_{\Sigma}^A$, then $\inver(\aut{A})=\inverse{\aut{A}}$. 
Apart from $\dz(\aut{A})$ which is always a Mealy automaton, the other
six sets may or may not define a Mealy automaton depending on
$\aut{A}$. 

\smallskip

By tracking the content of the sets of transitions, we observe the
following:
\begin{align*}
\bigl[ \dz\inver\dz\inver(\aut{A}) \ \in  \bigr. & \left. \trans_A^{\Sigma}\times
\trans_{\Sigma}^A \right]
 \implies  \\
& \left[ \aut{A},  \dz(\aut{A}), \inver(\aut{A}),  \dz\inver(\aut{A}),
 \inver\dz(\aut{A}),  \dz\inver\dz(\aut{A}),
 \inver\dz\inver(\aut{A}), \dz\inver\dz\inver(\aut{A}) \ \in  \perm_A^{\Sigma}\times
\perm_{\Sigma}^A \right] \:.
\end{align*}

Let us introduce some additional terminology. 

\begin{definition}
A Mealy automaton is  {\em invertible} if it belongs to $\trans_A^{\Sigma}\times
\perm_{\Sigma}^A$; and \emph{reversible} if it belongs to $\perm_A^{\Sigma}\times
 \trans_{\Sigma}^A$. A Mealy automaton is an {\em IR-automaton} if it
 is both invertible and  
 reversible, that is, if it belongs to $\perm_A^{\Sigma}\times
 \perm_{\Sigma}^A$. If $\dz\inver\dz\inver(\aut{A})$ is a Mealy
 automaton, we say that $\aut{A}$ (\resp $\dz(\aut{A}), \dots ,
\dz\inver\dz\inver(\aut{A})$) is {\em bireversible}.
\end{definition}

The terms ``invertible, reversible, and bireversible'' are standard since~\cite{mns}. 
The acronym IR-automaton is introduced for convenience. 
IR-automata are of particular
interest and the core of the paper is devoted to them. 
In Fig.~\ref{fi-Maut}, the right
Mealy automaton is an IR-automaton,
but not the left one. 

\paragraph{Mealy automaton of order $(n,k)$.}
Consider a Mealy automaton $\aut{A} = ( A,\Sigma,\delta,\rho)$ in
$\trans_A^{\Sigma}\times \trans_{\Sigma}^A$ and $n,k>0$. 
The quadruple 
\[
\aut{A}_{n,k} = \bigl( \ A^n,\Sigma^k, (\delta_{\mot{x}} : A^n \rightarrow A^n)_{\mot{x}\in
  \Sigma^k}, (\rho_{\mot{u}} : \Sigma^k \rightarrow \Sigma^k )_{\mot{u}\in
  A^n} \ \bigr) 
\]
is a Mealy automaton in
$\trans_{A^n}^{\Sigma^k}\times\trans_{\Sigma^k}^{A^n}$ that we call the
\emph{Mealy automaton of order $(n,k)$ 
associated with $\aut{A}$}. Observe that $\aut{A}_{1,1}=\aut{A}$.

In Fig.~\ref{fi-A21}, we show the Mealy automaton of order $(2,1)$
associated with the Mealy automaton of Fig.~\ref{fi-Mautdual}. 

\begin{figure}[ht]
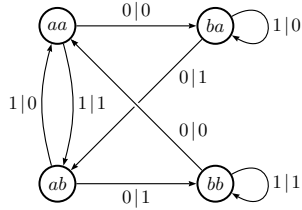

\begin{center}
\TinyPicture
\VCDraw{%
\begin{VCPicture}{(3,-.5)(11,5.5)}
\LargeState
\State[aa]{(4,5)}{C}
\State[ab]{(4,0)}{D}
\State[ba]{(9,5)}{E}
\State[bb]{(9,0)}{F}
\EdgeL[0.5]{C}{E}{\IOL{0}{0}}
\EdgeR[0.5]{D}{F}{\IOL{0}{1}}
\LoopE[0.5]{E}{\IOL{1}{0}}
\LoopE[0.5]{F}{\IOL{1}{1}}
\ArcL[0.5]{C}{D}{\IOL{1}{1}}
\ArcL[0.5]{D}{C}{\IOL{1}{0}}
\EdgeL[0.2]{E}{D}{\IOL{0}{1}}
\EdgeBorder
\EdgeR[0.2]{F}{C}{\IOL{0}{0}}
\EdgeBorderOff
\end{VCPicture}
}
\end{center}
\caption{Mealy automaton of order $(2,1)$.} 
\label{fi-A21} 
\end{figure}

\subsection{Helix graph}
We have already seen two
 equivalent ways of presenting a Mealy
automaton: $(i)$ as a quadruple $(A,\Sigma,\delta,\rho)$, $(ii)$ as a
labelled directed graph (see
Fig.~\ref{fi-Mautdual}).
We propose
here a third and original one which turns out to be very convenient. 

\smallskip

The \emph{helix graph} $\cH$ of a Mealy automaton $\aut{A}=(A,\Sigma,\delta,\rho)$ is the
directed graph with nodes \(A\times \Sigma\) and arcs \((x,i)
\longrightarrow \bigl(\delta_i(x), \rho_x(i)\bigr)\) for all \((x,i)\).
The \emph{helix graph $\cH_{n,k}$ of order $(n,k)$ associated with
  $\aut{A}$} is the helix graph of $\aut{A}_{n,k}$. 
In Fig.~\ref{fi-helix}, we have represented the helix graph of the
Mealy automaton of Fig.~\ref{fi-Mautdual}. 

\begin{figure}[ht]
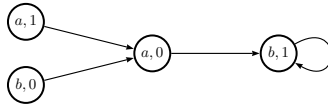

\begin{center}
\TinyPicture
\VCDraw{%
\begin{VCPicture}{(-1,4)(9,4.5)}
\LargeState
\ChgStateLabelScale{0.8}
\State[a,1]{(0,5)}{a1}
\State[b,0]{(0,3)}{b0}
\State[a,0]{(4,4)}{a0}
\State[b,1]{(8,4)}{b1}
\Edge{a1}{a0}
\Edge{b0}{a0}
\Edge{a0}{b1}
\LoopE{b1}{}
\end{VCPicture}
}
\end{center}
\caption{Helix graph.} 
\label{fi-helix} 
\end{figure}

Bireversible automata have a nice characterization using
the helix graph. 

\begin{lemma}\label{le-bi}
Consider an IR-automaton $\aut{A}$ with helix graph $\cH$. We have:
\[
\left[ \ \aut{A} \text{ bireversible } \right]  \iff  \left[ \ 
  \cH \text{ union of
    cycles } \right] \:.
\]
\end{lemma}

\begin{proof}
Define the directed graph $\widetilde{\cH}$ as follows: 
\begin{itemize}
\item nodes: \(A^{-1}\times \Sigma^{-1}\),
\item arcs: \((x^{-1},i^{-1})\longrightarrow (y^{-1},j^{-1})\) if
  \((y,j) \longrightarrow (x,i)\) is an arc of \(\cH\).
\end{itemize}
If $\dz\inver\dz\inver(\aut{A})$ is a Mealy automaton, then
$\widetilde{\cH}$ is its helix graph.

Conversely, assume that $\aut{A}$ is an IR-automaton and that
$\cH$ is a union of cycles. Consider a node $(y,j)$ of
$\cH$: it has a unique predecessor in $\cH_{1,1}$.
\end{proof}

\subsection{Automaton (semi)group}\label{sse-automgroup}
Let $\aut{A} = (A,\Sigma, \delta,\rho)$ be a Mealy automaton. 
We view $\aut{A}$ as an automaton with an input and an output tape, thus
defining mappings from input words over $\Sigma$ to output words
over~$\Sigma$. 
Formally, for $x\in A$, the map
$\rho_x : \Sigma^* \rightarrow \Sigma^*$,
extending $\rho_x : \Sigma \rightarrow \Sigma$, is defined by:
\[
\rho_x (\mot{u}) = \mot{v} \quad \textrm{if} \quad \exists y,
\ x\stackrel{\mot{u}|\mot{v}}{\longrightarrow} y \:.\]
By convention, the image of the empty word
is itself. The mapping $\rho_x$ is length-preserving and
prefix-preserving (the prefix of the image is the image of the prefix). 
It satisfies 
\begin{equation}\label{eq-property}
\forall u \in \Sigma, \ \forall \mot{v} \in \Sigma^*, \qquad
\rho_x(u\mot{v}) = \rho_x(u)\rho_{\delta_u(x)}(\mot{v}) \:.
\end{equation} 
We can also use \eref{eq-property} to define 
$\rho_x:\Sigma^* \rightarrow \Sigma^*$ inductively starting from $\rho_x:\Sigma
\rightarrow \Sigma$. We say that $\rho_x$ is the \emph{production
function} associated with
$(\aut{A},x)$. 
For $\mot{u}=u_1\cdots u_n \in A^n$, $n>0$, set
\(\rho_\mot{u}: \Sigma^* \rightarrow \Sigma^*, \rho_\mot{u} = \rho_{u_n}
\circ \cdots \circ \rho_{u_1} \:\).

\begin{definition}
Consider $\aut{A} \in \trans_A^{\Sigma}\times \trans_{\Sigma}^A$. 
The semigroup of mappings from $\Sigma^*$ to $\Sigma^*$ generated by
$\rho_x, x\in A$, is called the \emph{semigroup of $\aut{A}$} and is
denoted by $\presm{\aut{A}}$. 
Assume that $\aut{A} \in \trans_A^{\Sigma}\times \perm_{\Sigma}^A$. The group of
mappings from $\Sigma^*$ to $\Sigma^*$ generated as a group by
$\rho_x, x\in A$, is called the \emph{group of $\aut{A}$} and is 
denoted by $\pres{\aut{A}}$.
\end{definition}

The above definition makes sense. Indeed if $\aut{A} \in
\trans_A^{\Sigma}\times \perm_{\Sigma}^A$, then 
the production mapping
$\rho_x$ associated with $(\aut{A},x)$ is a bijection from $\Sigma^*$ to $\Sigma^*$.
The inverse
bijection $\rho_x^{-1}:\Sigma^*\rightarrow \Sigma^*$ is the
production mapping $\rho_{x^{-1}}$ associated with $(\inverse{\aut{A}},x^{-1})$, where $\inverse{\aut{A}}$
is the inverse Mealy automaton defined in \eref{eq-inv}. Therefore, we have 
 \[
 \presm{\aut{A}}= \{ \rho_\mot{u}, \mot{u} \in A^* \}, \qquad
 \pres{\aut{A}} = \{ \rho_\mot{u}, \mot{u} \in (A\sqcup A^{-1})^* \}
 \:.
 \]

\begin{lemma}\label{lm-g-sg-finis}
Let \(\aut{A}\) be an IR-automaton.
Then we have
$\pres{\aut{A}}=\pres{\inverse{\aut{A}}} = \pres{ \aut{A} \sqcup
  \inverse{\aut{A}}} = \presm{ \aut{A} \sqcup \inverse{\aut{A}}}$, where
$\aut{A} \sqcup \inverse{\aut{A}}$ is the Mealy automaton whose set of transitions
is the union of the ones of 
$\aut{A}$ and $\inverse{\aut{A}}$. 
Furthermore, if either \(\pres{\aut{A}}\) or \(\presm{\aut{A}}\) is
finite, then we have~\(\pres{\aut{A}} = \presm{\aut{A}}\).
\end{lemma}

\begin{proof}
The first statement follows directly from the definitions. 
Suppose that \(\presm{\aut{A}}\) is finite and let \(x\) be one of its
elements. Since the semigroup  \(\presm{\aut{A}}\) is finite, there exist~$k$ and~$n$ such that $x^{n+k}=x^k$. So we have \(x^n=1\) in the group \(\pres{\aut{A}}\). 
 Hence the inverse of \(x\) is
\(x^{n-1}\) which belongs to the semigroup \(\presm{\aut{A}}\). So we
have  \(\pres{\aut{A}} = \presm{\aut{A}}\). 
Assume now that \(\pres{\aut{A}}\) is finite.
Since the semigroup~\(\presm{\aut{A}}\) naturally embeds into the group~\(\pres{\aut{A}}\), it is also
finite. 
\end{proof}

\begin{definition}\label{de-gag}
A semigroup $M$ is called an \emph{automaton semigroup} if there exists
a Mealy automaton $\aut{A}$ such that $M = \presm{\aut{A}}$. 
A group $G$ is
called an \emph{automaton group} if there exists an invertible Mealy automaton
$\aut{A}$ such that $G = \pres{\aut{A}}$.  
In both cases, we say that $\aut{A}$ \emph{generates} the (semi)group.
\end{definition}

Denote dually by $\delta_i:A^*\rightarrow A^*,
i\in \Sigma$, the production mappings associated with
the dual Mealy automaton $\dz(\aut{A})$. For $\mot{v}=v_1\cdots v_n
\in \Sigma^n$, $n>0$, set $\delta_\mot{v}: A^* \rightarrow A^*,
\ \delta_\mot{v} = \delta_{v_n}\circ \cdots \circ \delta_{v_1}$. 

\smallskip

A pair of Mealy automata $\{\aut{A},\dz(\aut{A})\}$ generates a pair
of (semi)groups.

\medskip

Examples of automata (semi)groups are given in Table~\ref{tbl-examples}.

\bigskip

The two following propositions complement each other.
Proposition~\ref{pr:duale-finitude} is proved by Nekrashevych for a pair
of dual
bireversible Mealy automata~\cite[Lem.1.10.6]{nek}. For the sake of
completeness, we provide a similar proof in the general case. 

\begin{proposition}\label{pr-fifi}
Let $G$ and $H$ be two finite semigroups. There exists a Mealy automa\-ton
$\aut{A}$ such that
$\pres{\aut{A}}_+ =G$ and $\pres{\dz(\aut{A})}_+ =H$.
Let $G$ and $H$ be two finite groups. There exists an IR-automa\-ton
$\aut{A}$ such that
$\pres{\aut{A}} =G$ and $\pres{\dz(\aut{A})} =H$.
\end{proposition}

\begin{proof}
We carry out the proof for groups. The argument is similar for
semigroups. 

Any finite group is a subgroup of a permutation group. Let $\Sigma_1$
and $A_2$ be two finite sets such that $G$ is a subgroup of
$\perm_{\Sigma_1}$ and $H$ is a subgroup of $\perm_{A_2}$. Let
$A_1 \subset \perm_{\Sigma_1}$ be a set of generators of $G$, let 
$\Sigma_2 \subset \perm_{A_2}$ be a set of generators of $H$.

Set $A = A_1\times A_2$ and $\Sigma=\Sigma_1\times \Sigma_2$. Consider
the Mealy automaton $\aut{A}$ with states $A$, alphabet $\Sigma$, and
transitions
\[
(a,b) 
\xrightarrow{(i,j) \mid (a(i),j)}
(a,j(b)) \:.
\]
Denote the corresponding mappings by $\delta$ and $\rho$. 
Clearly, for $(a,b)\in A_1\times A_2$ and $(a,b')\in A_1\times A_2$,
we have $\rho_{(a,b)}=\rho_{(a,b')}$ and we denote this mapping by
$\rho_a:\Sigma^*\rightarrow \Sigma^*$. We have, $\forall a\in A_1,
\forall (i_1,j_1)\cdots (i_n,j_n) \in \Sigma^*,$ 
\[
\rho_a \bigl( (i_1,j_1)\cdots (i_n,j_n) \bigr) =
(a(i_1),j_1) \ (a(i_2),j_2)\ \cdots \ (a(i_n),j_n) \:.
\]
So the group generated by $(\rho_a:\Sigma^*\rightarrow \Sigma^*)_{a\in
  A_1}$ is isomorphic to the group generated by
$(a:\Sigma_1\rightarrow \Sigma_1)_{a\in A_1}$. That is $\pres{\aut{A}} =
G$. 
Similarly, $\pres{\dz(\aut{A})} = H$. 
\end{proof}

\begin{sidewaystable}
\centering
\caption{Examples of automata (semi)groups.\label{tbl-examples}}
{
\begin{tabular}{|m{98pt}|m{75pt}|m{70pt}|c|c|c|m{70pt}|m{70pt}|m{65pt}|}
\hline
\multicolumn{3}{|c|}{infinite world}
&\!\!\multirow{2}{*}{\rotatebox{270}{\footnotesize \!\!\!invertible}}\!\!
&\!\!\multirow{2}{*}{\rotatebox{270}{\footnotesize \!\!\!reversible}}\!\!
&\!\!\multirow{2}{*}{\rotatebox{270}{\footnotesize \!\!\!bireversible}}\!\!
&\multicolumn{3}{c|}{finite world}\\
\cline{1-3}\cline{7-9}
\multicolumn{1}{|c|}{\small generated (semi)group}
&\multicolumn{1}{c|}{diagram}
&\multicolumn{1}{c|}{helix graph}
&&&&\multicolumn{1}{c|}{helix graph}
&\multicolumn{1}{c|}{diagram}
&\multicolumn{1}{c|}{\small gen. (semi)group}\\
\cline{1-3}\cline{7-9}
\footnotesize
the semigroup~$\mathbf{S_{I_2}}$

(the very smallest Mealy automaton with intermediate growth, see~\cite{brs})
&\centering
\SmallPicture\VCDraw{%
\begin{VCPicture}{(0,-.7)(4,2.3)}
\State[a]{(0,0)}{A} \State[b]{(4,0)}{B}
\EdgeL{B}{A}{\IOL{1}{1}}
\LoopN[.7]{A}{\StackTwoLabels{\IOL{0}{1}}{\IOL{1}{0}}}
\LoopN[.4]{B}{\IOL{0}{1}}
\end{VCPicture}}
&\centering
\SmallPicture\VCDraw{%
\begin{VCPicture}{(0,-.7)(2,2.7)}
\State[a0]{(0,0)}{A0} \State[a1]{(2,0)}{A1}
\State[b1]{(2,2)}{B1} \State[b0]{(0,2)}{B0}
\ArcR{A0}{A1}{}
\ArcR{A1}{A0}{}
\EdgeL{B1}{A1}{}
\EdgeL{B0}{B1}{}
\end{VCPicture}}
&&&
&\centering
\SmallPicture\VCDraw{%
\begin{VCPicture}{(-2,-0.3)(2,2.3)}
\State[a0]{(-2,2)}{A0}	\State[a2]{(0,2)}{A2} 	\State[b2]{(2,2)}{B2}
\State[b0]{(-2,0)}{B0}	\State[b1]{(0,0)}{B1} 	\State[a1]{(2,0)}{A1}
\EdgeL{A0}{B0}{}
\EdgeL{A2}{B0}{}
\EdgeL{B2}{A2}{}
\EdgeL{A1}{B2}{}
\ArcR{B0}{B1}{}
\ArcR{B1}{B0}{}
\end{VCPicture}}
&\centering
\SmallPicture\VCDraw{%
\begin{VCPicture}{(0,-1.3)(4,2.3)}
\State[a]{(0,0)}{A} \State[b]{(4,0)}{B}
\CLoopN[.3]{B}{\StackTwoLabels{\IOL{0}{1}}{\IOL{1}{0}}}
\ArcL[.2]{B}{A}{\IOL{2}{2}}
\ArcL[.2]{A}{B}{\StackThreeLabels{\IOL{0}{0}}{\IOL{1}{2}}{\IOL{2}{0}}}
\end{VCPicture}}
&
an order~13597

semigroup
\\
\hline

\footnotesize
the {\bf Grigorchuk group}

see~\cite{gns}
&\centering
\TinyPicture\VCDraw{%
\begin{VCPicture}{(0,-4.8)(4,0.8)}
\State[a]{(0,-4)}{A} \State[b]{(0,0)}{B} \State[c]{(2,-2)}{C} \State[d]{(4,0)}{D} \State[e]{(4,-4)}{E}
\EdgeL[.8]{A}{E}{\StackTwoLabels{\IOL{0}{1}}{\IOL{1}{0}}}
\EdgeR[.3]{B}{A}{\IOL{0}{0}}
\EdgeR[.7]{B}{C}{\IOL{1}{1}}
\EdgeL[.3]{C}{A}{\IOL{0}{0}}
\EdgeR[.3]{C}{D}{\IOL{1}{1}}
\EdgeL[.3]{D}{E}{\IOL{0}{0}}
\EdgeL{D}{B}{\IOL{1}{1}}
\CLoopE[.2]{E}{\StackTwoLabels{\IOL{0}{0}}{\IOL{1}{1}}}
\end{VCPicture}}
&\centering
\TinyPicture\VCDraw{%
\begin{VCPicture}{(1,-2.3)(5,2.3)}
\State[a0]{(1,0)}{a0} \State[a1]{(5,0)}{a1}
\State[b0]{(0.4,2)}{b0} \State[b1]{(5,2)}{b1}
\State[c0]{(1.6,2)}{c0} \State[c1]{(3,2)}{c1}
\State[d0]{(3,-2)}{d0} \State[d1]{(3,0)}{d1}
\State[e0]{(5,-2)}{e0} \State[e1]{(1,-2)}{e1}
\EdgeL{b0}{a0}{}
\EdgeL{c0}{a0}{}
\EdgeL{a0}{e1}{}
\CLoopNE{e1}{}
\EdgeL{d0}{e0}{}
\EdgeL{a1}{e0}{}
\CLoopNW{e0}{}
\EdgeL{b1}{c1}{}
\EdgeL{c1}{d1}{}
\EdgeR{d1}{b1}{}
\end{VCPicture}}
&\cellcolor{gray}&&
&\centering \SmallPicture\VCDraw{%
\begin{VCPicture}{(-2.5,-0.3)(2.5,2.3)}
\State[c0]{(-2,2)}{c0}	\State[b0]{(0,2)}{b0}	\State[c1]{(2,2)}{c1}
\State[a0]{(-2,0)}{a0}\State[a1]{(0,0)}{a1} \State[b1]{(2,0)}{b1}
\EdgeL{c1}{b1}{}
\EdgeL{b0}{a0}{}
\EdgeL{c0}{a0}{}
\EdgeL{b1}{a1}{}
\ArcR{a0}{a1}{}
\ArcR{a1}{a0}{}
\end{VCPicture}}
&\centering\SmallPicture\VCDraw{%
\begin{VCPicture}{(0,-.3)(4,2.3)}
\State[a]{(2,0)}{A} \State[b]{(0,2)}{B} \State[c]{(4,2)}{C}
\CLoopE[.55]{A}{\StackTwoLabels{\IOL{0}{1}}{\IOL{1}{0}}}
\EdgeR{B}{A}{\StackTwoLabels{\IOL{0}{0}}{\IOL{1}{1}}}
\EdgeL[.3]{C}{A}{{\IOL{0}{0}}}
\EdgeR{C}{B}{{\IOL{1}{1}}}
\end{VCPicture}}
&
the group~$\Z_2\times D_4$\\
\hline

\footnotesize
the {\bf Basilica group}

see~\cite{gns}
&\centering
\SmallPicture\VCDraw{%
\begin{VCPicture}{(0,-2.8)(4,1.2)}
\State[a]{(0,0)}{A} \State[b]{(4,0)}{B} \State[c]{(2,-2)}{C}
\ArcL[.2]{A}{B}{\IOL{0}{1}}
\EdgeR[.25]{A}{C}{\IOL{1}{0}}
\ArcL[.5]{B}{A}{\IOL{0}{0}}
\EdgeL[.25]{B}{C}{\IOL{1}{1}}
\CLoopE[.51]{C}{\StackTwoLabels{\IOL{0}{0}}{\IOL{1}{1}}}
\end{VCPicture}}
&\centering
\SmallPicture\VCDraw{%
\begin{VCPicture}{(-2.5,-0.3)(2.5,2.3)}
\State[c0]{(-2,2)}{c0}		\State[a1]{(0,2)}{a1} 	\State[c1]{(2,2)}{c1}
\State[b0]{(-2,0)}{b0}	\State[a0]{(0,0)}{a0} 	\State[b1]{(2,0)}{b1}
\EdgeL{b0}{a0}{}
\EdgeL{a0}{b1}{}
\EdgeL{a1}{c0}{}
\EdgeL{b1}{c1}{}
\CLoopS{c0}{}
\CLoopW{c1}{}
\end{VCPicture}}
&\cellcolor{gray}&&
&\centering
\SmallPicture\VCDraw{%
\begin{VCPicture}{(0,0)(2,2)}
\State[b0]{(0,2)}{b0} \State[a0]{(2,2)}{a0} 
\State[b1]{(0,0)}{b1}\State[a1]{(2,0)}{a1}
\EdgeL{b0}{a0}{}
\EdgeL{b1}{a1}{}
\ArcR{a0}{a1}{}
\ArcR{a1}{a0}{}
\end{VCPicture}}
&\centering
\SmallPicture\VCDraw{%
\begin{VCPicture}{(0,-1)(4,1)}
\State[a]{(0,0)}{A} \State[b]{(4,0)}{B}
\LoopN[.8]{A}{\StackTwoLabels{\IOL{0}{1}}{\IOL{1}{0}}}
\EdgeL{B}{A}{\StackTwoLabels{\IOL{0}{0}}{\IOL{1}{1}}}
\end{VCPicture}}
&the Klein $4$-group $V=\Z_2\times\Z_2$
\\
\hline

\footnotesize
the {\bf lamplighter group}

$L=\Z\wr\Z_2$

see~\cite{gns}
&\centering \SmallPicture\VCDraw{%
\begin{VCPicture}{(0,-1.3)(4,2)}
\State[a]{(0,0)}{A} \State[b]{(4,0)}{B}
\ArcL{A}{B}{\IOL{0}{1}}
\ArcL{B}{A}{\IOL{0}{0}}
\CLoopN[.4]{A}{\IOL{1}{0}}
\CLoopN[.6]{B}{\IOL{1}{1}}
\end{VCPicture}}
&\centering \SmallPicture\VCDraw{%
\begin{VCPicture}{(-.5,-1.1)(4.5,1.1)}
\State[a0]{(2,0)}{a0} \State[a1]{(0,.8)}{a1}
\State[b0]{(0,-.8)}{b0} \State[b1]{(4,0)}{b1}
\EdgeL{a1}{a0}{}
\EdgeL{b0}{a0}{}
\EdgeL{a0}{b1}{}
\CLoopN{b1}{}
\end{VCPicture}}
&\cellcolor{gray}&\cellcolor{gray}
&\\
\hline

{\footnotesize the rank~3 free group

({\bf Ale\v{s}in} automaton)

see~\cite{aleshin,svv}}
&\centering
\SmallPicture\VCDraw{%
\begin{VCPicture}{(0,-2.5)(4,1.2)}
\State[a]{(0,0)}{A} \State[b]{(2,-2)}{B} \State[c]{(4,0)}{C}
\ArcR[.5]{A}{C}{\IOL{0}{1}}
\EdgeR[.2]{A}{B}{\IOL{1}{0}}
\CLoopW[.5]{B}{\IOL{0}{1}}
\EdgeR[.2]{B}{C}{\IOL{1}{0}}
\ArcR[.1]{C}{A}{\StackTwoLabels{\IOL{0}{0}}{\IOL{1}{1}}}
\end{VCPicture}}
&\centering \SmallPicture\VCDraw{%
\begin{VCPicture}{(-.5,-1.3)(4.5,1.3)}
\State[b1]{(0,1)}{b1}		\State[c0]{(2,1)}{c0}		\State[a0]{(4,1)}{a0}
\State[b0]{(0,-1)}{b0}	\State[a1]{(2,-1)}{a1}	\State[c1]{(4,-1)}{c1}
\EdgeL{b1}{c0}{}
\EdgeL{c0}{a0}{}
\EdgeL{a0}{c1}{}
\EdgeL{c1}{a1}{}
\EdgeL{a1}{b0}{}
\EdgeL{b0}{b1}{}
\end{VCPicture}}
&\cellcolor{gray}&\cellcolor{gray}&\cellcolor{gray}
&\centering \SmallPicture\VCDraw{%
\begin{VCPicture}{(-.5,-1.3)(4.5,1.3)}
\State[a0]{(0,1)}{a0} 	\State[a1]{(2,1)}{a1}		\State[b0]{(4,1)}{b0}
\State[b1]{(0,-1)}{b1}	\State[b2]{(2,-1)}{b2}	\State[a2]{(4,-1)}{a2}
\EdgeL{b0}{a1}{}
\EdgeL{a1}{b2}{}
\EdgeL{b2}{a2}{}
\EdgeR{a2}{b0}{}
\ArcR{a0}{b1}{}
\ArcR{b1}{a0}{}
\end{VCPicture}}
&\centering\SmallPicture\VCDraw{%
\begin{VCPicture}{(0,-2.3)(4,2.3)}
\State[a]{(0,0)}{A} \State[b]{(4,0)}{B}
\ArcL[.1]{A}{B}{\StackThreeLabels{\IOL{0}{1}}{\IOL{1}{2}}{\IOL{2}{0}}}
\ArcL[.1]{B}{A}{\StackThreeLabels{\IOL{0}{1}}{\IOL{1}{0}}{\IOL{2}{2}}}
\end{VCPicture}}
&an order~36 group\\
\hline

{\footnotesize
the free product

$\Z_2^{*3}=\Z_2*\Z_2*\Z_2$

({\bf BabyAle\v{s}in} automaton)

see~\cite{svv}}
&\centering
\SmallPicture\VCDraw{%
\begin{VCPicture}{(0,-2.5)(4,1.2)}
\State[a]{(0,0)}{A} \State[b]{(2,-2)}{B} \State[c]{(4,0)}{C}
\ArcL[.1]{A}{C}{\StackTwoLabels{\IOL{0}{1}}{\IOL{1}{0}}}
\EdgeL[.2]{B}{A}{\IOL{0}{0}}
\CLoopE[.5]{B}{\IOL{1}{1}}
\EdgeL[.2]{C}{B}{\IOL{0}{0}}
\ArcL[.5]{C}{A}{\IOL{1}{1}}
\end{VCPicture}}
&\centering\SmallPicture\VCDraw{%
\begin{VCPicture}{(-.5,-1.3)(4.5,1.3)}
\State[c0]{(0,1)}{c0}		\State[a1]{(2,1)}{a1}		\State[b1]{(4,1)}{b1}
\State[b0]{(0,-1)}{b0}	\State[a0]{(2,-1)}{a0} 	\State[c1]{(4,-1)}{c1}
\EdgeL{c0}{b0}{}
\EdgeL{b0}{a0}{}
\EdgeL{a0}{c1}{}
\EdgeL{c1}{a1}{}
\EdgeL{a1}{c0}{}
\CLoopS{b1}{}
\end{VCPicture}}
&
\cellcolor{gray}
&
\cellcolor{gray}
&
\cellcolor{gray}

&
\centering \SmallPicture\VCDraw{%
\begin{VCPicture}{(0.5,-1.7)(5.5,1.7)}
\State[a3]{(1,1.5)}{a3}	\State[b0]{(3,1.5)}{b0}	\State[b1]{(5,1.5)}{b1}
\State[a2]{(1,0)}{a2} 						\State[a0]{(5,0)}{a0}
\State[b3]{(1,-1.5)}{b3}	\State[b2]{(3,-1.5)}{b2} 	\State[a1]{(5,-1.5)}{a1}
\EdgeR{a2}{b3}{}
\EdgeR{b3}{b2}{}
\EdgeR{b2}{a1}{}
\EdgeR{a1}{a0}{}
\EdgeR{a0}{b1}{}
\EdgeR{b1}{b0}{}
\EdgeR{b0}{a3}{}
\EdgeR{a3}{a2}{}
\end{VCPicture}}
&
\centering\SmallPicture\VCDraw{%
\begin{VCPicture}{(0,-1.9)(4,2.5)}
\State[a]{(0,0)}{A} \State[b]{(4,0)}{B}
\ArcL{A}{B}{\StackTwoLabels{\IOL{0}{1}}{\IOL{2}{3}}}
\ArcL{B}{A}{\StackTwoLabels{\IOL{0}{3}}{\IOL{2}{1}}}
\CLoopN[.5]{A}{\StackTwoLabels{\IOL{1}{0}}{\IOL{3}{2}}}
\CLoopN[.5]{B}{\StackTwoLabels{\IOL{1}{0}}{\IOL{3}{2}}}
\end{VCPicture}}
&
the group~$G_{16}^{(9)}$
\\
\hline
\end{tabular}}
\end{sidewaystable}


\begin{proposition}\label{pr:duale-finitude}
Let $\aut{A}$ be a Mealy automaton. The semigroup ~$\pres {\aut{A}}_+$ is  finite if
and only if the semigroup~$\pres{\dz(\aut{A})}_+$ is finite.
\end{proposition}

Proposition \ref{pr:duale-finitude} extends to groups using
Lemma~\ref{lm-g-sg-finis}.

 \begin{proof}
 Set $\aut{A}=(A,\Sigma,\delta,\rho)$ and 
 assume that
 $\pres{\dz(\aut{A})}_+= \{\delta_{\mot{u}}: A^* \rightarrow A^*, \ \mot{u}\in
 \Sigma^*\}$ is finite.
Consider the Cayley graph~$\cal G$ of~$\pres{\dz(\aut{A})}_+$ with respect to the set of generators~$\Sigma$, see the left of the figure just below. Now fix $\mot{w}\in
A^*$ and recall that 
\[
\rho_{\mot{w}}(u_1u_2\cdots u_n) := 
 \rho_{\mot{w}}(u_1)\rho_{\delta_{u_1}(\mot{w})}(u_2)\rho_{\delta_{u_1u_2}(\mot{w})}(u_3)\cdots
 \rho_{\delta_{u_1u_2\cdots u_{n-1}}(\mot{w})}(u_n)\:,
\]
for all $u_1u_2\cdots u_n \in \Sigma^*$. 
This shows that $\rho_{\mot{w}}$ can also be described as the output map of a
 letter-to-letter transducer built upon~$\cal G$, see the right of the figure. 
\begin{center}
\TinyPicture
\VCDraw{%
\begin{VCPicture}{(-6,-0.2)(10,0)}
\LargeState
\State[\delta_{\mot{u}}]{(-5,0)}{A}
\State[\delta_{\mot{u}i}]{(0,0)}{B}
\State[\delta_{\mot{u}}]{(4,0)}{C}
\State[\delta_{\mot{u}i}]{(9,0)}{D}
\EdgeL[.5]{C}{D}{\IOL{i}{\rho_{\delta_{\mot{u}}(\mot{w})}(i)}}
\ChgEdgeLabelSep{2}
\EdgeL[.5]{A}{B}{i}
\end{VCPicture}
}
\end{center}

\noindent Now observe that there is only a finite number of possible
different transducers built on $\cal G$, which is equal to the number
of different mappings from $\pres{\dz(  \aut{A})}_+$ to~$\trans_{\Sigma}$. We conclude that $\#
\pres{\aut{A}}_+ \leq \bigl(\# \Sigma \bigr)^{(\#
    \Sigma)\ ( \# \pres{\dz(  \aut{A})}_+)}$.  
\end{proof}

The growth of a Mealy automaton is defined as the growth of the number
of different elements $\rho_{\mot{u}}, \ \mot{u} \in A^n$, as a
function of $n$, see~\cite{brs,growth}. Automata generating finite (semi)groups are those of
finite growth. Looking at the 2-letter
2-state automata, it appears that it is the only growth class within
the known growth classes (finite, polynomial, intermediate and exponential) to
be stable by dualization.

\medskip

Let \(\aut{A}\) be an IR-automaton. Recall that~$\pres{\aut{A}} =
\pres{\aut{A}\sqcup \aut{A}^{-1}}$. In words, considering the states and
their inverses does not modify the generated group. 
We can also
consider the letters and their inverses. Set \(\widetilde{\aut{A}}=
\aut{A}' \sqcup (\aut{A}')^{-1}\) where  \(\aut{A}'=\dz(\dz(\aut{A})
\sqcup \inverse{\dz(\aut{A})})\). The Mealy automaton
\(\widetilde{\aut{A}}\) is the extension of \(\aut{A}\) with stateset
\(\alphA\sqcup\inverse{\alphA}\) and alphabet
\(\alphS\sqcup\inverse{\alphS}\).

Next result is a corollary of Proposition \ref{pr:duale-finitude} and
Lemma~\ref{lm-g-sg-finis}.

\begin{corollary}\label{cor-gen}
Let \(\aut{A}\) be an IR-automaton. The groups
\(\pres{\aut{A}}\) and
\(\pres{\tilde{\aut{A}}}\) are either both finite or both infinite.
\end{corollary}

The above groups are not necessary equal. Consider for instance the
automaton \(\aut{A}\) generating \(G_{16}^{(9)}\) in
Table~\ref{tbl-examples}: we have  \(|\pres{\aut{A}}|=16\) and
\(|\pres{\widetilde{\aut{A}}}|=64\).

\section{Reduction of Mealy automata and finiteness}
\label{s:reduction}

Here we define the \emph{$\mz\dz$-reduction} of Mealy automata which
provides a sufficient condition of finiteness. The condition is not
necessary and two counterexamples are provided. 

\subsection{Minimization of a Mealy automaton}\label{quo-auto}
\begin{definition}
Let $\Ac=(A,\Sigma,\delta,\rho)$ be a Mealy automaton.
An equivalence $\equiv$ on $A$ is a \emph{congruence} for $\Ac$ if
\[
\left[\forall x,y\in A,\ x\equiv y\right] \Longrightarrow
\left[\forall i\in\Sigma,\ \rho_x(i)=\rho_y(i)\text{ and }
\delta_i(x)\equiv\delta_i(y)\right].
\]
The \emph{Nerode equivalence} on $A$ is the coarsest congruence for $\Ac$.
\end{definition}
The Nerode equivalence is the limit of the sequence $(\equiv_k)$ of
increasingly finer equivalences defined recursively by:
\begin{align*}
\forall x,y\in A,\qquad\qquad x\equiv_0 y & \Longleftrightarrow \forall i\in\Sigma,\ \rho_x(i)=\rho_y(i),\\
\forall k\geqslant 0, x\equiv_{k+1} y & \Longleftrightarrow x\equiv_k y
 \text{ and }\forall i\in\Sigma,\ \delta_i(x)\equiv_k\delta_i(y).
\end{align*}
Since the set $A$ is finite, this sequence is ultimately constant; moreover if
two consecutive equivalences are equal, the sequence remains constant from this
point. The limit is therefore computable.
For every $x$ in $A$, we denote by $[x]$ the class of $x$ w.r.t. the Nerode equivalence.

\begin{definition}
Let $\Ac=(A,\Sigma,\delta,\rho)$ be a Mealy automaton
and let $\equiv$ be the Nerode equivalence on  $\Ac$.
The \emph{minimization} of $\Ac$ is the Mealy automaton
\(\Ac/\negthickspace\equiv\,=(A/\negthickspace\equiv,\Sigma,\tilde{\delta},\tilde{\rho})\),
where for every $(x,i)$ in $A\times \Sigma$,
$\tilde{\delta}_i([x])=[\delta_i(x)]$ and
$\tilde{\rho}_{[x]}(i)=\rho_x(i)$.
\end{definition}

This definition is consistent with the minimization of ``deterministic
finite automata'', where instead of considering the production functions $(\rho_x)_x$, the computation of the congruence is
initiated by the separation between terminal and non-terminal states.

\begin{lemma}\label{lem-min}
Let $\Ac=(A,\Sigma,\delta,\rho)$ be a Mealy automaton,
and let $\Ac/\negthickspace\equiv$ be its minimization.
The function on $\Sigma^*$ generated by~$x$ in~$\Ac$
is equal to the function generated by~$[x]$ in~$\Ac/\equiv$.
Therefore, the Mealy automata $\Ac$ and $\Ac/\equiv$ generate the same
semigroup.
\end{lemma}

\begin{proof}
Let $(\mot{i}_n)_{n\in\N}$ be a sequence of words of
$\Sigma^*$ such that for all integer \(n\), the length of
\(\mot{i}_n\) is \(n\) and \(\mot{i}_n\) is a prefix of
\(\mot{i}_{n+1}\): \(\mot{i}_{n+1} = \mot{i}_ni_{n+1}\), where
\(i_{n+1}\in\Sigma\). We prove by induction on $n$ that for every
 $x$ of $A$, we have $\rho_x = \tilde{\rho}_{[x]}$ on $\Sigma^n$. 
It is obviously true
for $n=0$. If $n>0$:
\begin{align*}
\rho_x(\mot{i}_n)=&\rho_x(\mot{i}_{n-1})\rho_{\delta_{\mot{i}_{n-1}}(x)}(i_n)\\
=&\tilde{\rho}_{[x]}(\mot{i}_{n-1})\tilde{\rho}_{[\delta_{\mot{i}_{n-1}}(x)]}(i_n)\\
=&\tilde{\rho}_{[x]}(\mot{i}_{n-1})\tilde{\rho}_{\tilde{\delta}_{\mot{i}_{n-1}}([x])}(i_n)
=\tilde{\rho}_{[x]}(\mot{i}_n).
\end{align*}
\end{proof}

\subsection{The $\mz\dz$-reduction of Mealy automata}

Observe that the minimization of a Mealy automaton with a minimal dual
can make the dual automaton non-minimal.

\begin{definition}
A pair of dual Mealy automata is \emph{reduced} if both Mealy automata
are minimal. 
Let~$\mz$ be the operation of minimization; recall that $\dz$ is the
operation of dualization. The \emph{$\mz\dz$-reduction} of a Mealy
automaton consists in minimizing the automaton or its dual until the
resulting pair of dual Mealy automata is reduced. 
\end{definition}

If both a Mealy automaton and its dual automaton are non-minimal,
the procedure of $\mz\dz$-reduction seems to be dependent on the first
automaton chosen for the minimization. The reduction is actually confluent:

\begin{proposition}\label{prop:paire}
If $(\Ac,\Bc)$ is a pair of dual Mealy automata, the reduced pair
obtained by minimizing $\Ac$ first  is the same as the one obtained by
minimizing $\Bc$ first.
\end{proposition}

\begin{proof}
If $(\Ac,\Bc)$ is reduced, both Mealy automata are minimal,
and the proposition trivially holds.

Otherwise, the proof is by induction on the total number of states in
$\Ac$ and~$\Bc$. Let $(\Ac_1,\Bc_1)$ be the pair obtained by minimizing
$\Ac$
and let $(\Ac_2,\Bc_2)$ be the pair obtained by minimizing
$\Bc$.
Let us set $\Ac=(A,\Sigma,\delta,\rho)$,
$\Ac_1=(A_1,\Sigma,\delta^{(1)},\rho^{(1)})$,
and $\Ac_2=(A,\Sigma_2,\delta^{(2)},\rho^{(2)})$.
Let $\equiv_1$ and $\equiv_2$ be the congruences on $\Ac$ 
and $\Bc$ such that $A_1=A/\negthickspace\equiv_1$ and $\Sigma_2=\Sigma/\negthickspace\equiv_2$.
We show that $\equiv_1$ is a congruence on $\Ac_2$.
Let $x$ and $y$ be in $A$ such that $x\equiv_1 y$.
Then, for every $i$ in $\Sigma$, $\rho_x(i)=\rho_y(i)$ and
therefore, $\rho^{(2)}_x([i])=[\rho_x(i)]=[\rho_y(i)]=\rho^{(2)}_y([i])$;
besides,
$\delta^{(2)}_{[i]}(x)=\delta_i(x)\equiv_1\delta_i(y)=\delta^{(2)}_{[i]}(y)$.
Hence, $\equiv_1$ is a congruence on $\Ac_2$ and, likewise,
$\equiv_2$ is a congruence on $\Bc_1$.
We consider now the Mealy automaton
$\Ac'=(A_1,\Sigma_2,\delta',\rho')$ which is the quotient of $\Ac_2$
with respect to $\equiv_1$, and $\Bc'=(\Sigma_2,A_1,\rho",\delta")$
which is the quotient of $\Bc_1$ w.r.t. $\equiv_2$.
For every $x$ in $A$ and every $i$ in $\Sigma$,
it holds:
\[
\delta"_{[i]}([x])=\delta^{(1)}_i([x])=
[\delta_i(x)]=[\delta^{(2)}_{[i]}(x)]=\delta'_{[i]}([x])\:.
\]
Thus, $\delta"=\delta'$ and likewise $\rho"=\rho'$.

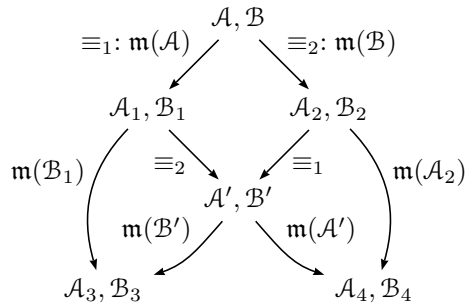
\begin{center}
\VCDraw{\begin{VCPicture}{(-3.5,-3.4)(3.5,3.3)}%
\ChgStateLineStyle{none}
\StateVar[\Ac,\Bc]{(0,3)}{A}
\StateVar[\Ac_1,\Bc_1]{(-2,1)}{A1}
\StateVar[\Ac_2,\Bc_2]{(2,1)}{A2}
\StateVar[\Ac',\Bc']{(0,-1)}{AP}
\StateVar[\Ac_3,\Bc_3]{(-3,-3)}{A3}
\StateVar[\Ac_4,\Bc_4]{(3,-3)}{A4}
\EdgeR{A}{A1}{\equiv_1:\mz(\Ac)}
\EdgeL{A}{A2}{\equiv_2:\mz(\Bc)}
\EdgeR{A1}{AP}{\equiv_2}
\EdgeL{A2}{AP}{\equiv_1}
\ArcL{AP}{A3}{}\LabelR{\mz(\Bc')}
\ArcR{AP}{A4}{}\LabelL{\mz(\Ac')}
\LArcR{A1}{A3}{\mz(\Bc_1)}
\LArcL{A2}{A4}{\mz(\Ac_2)}
\end{VCPicture}}
\end{center}

\noindent Consider now $\Ac_2=(A,\Sigma_2, \delta^{(2)}, \rho^{(2)})$ and
$\Ac'=(A_1,\Sigma_2, \delta', \rho')$. Clearly, applying the coarsest
congruences respectively on $A$ in \(\Ac_2\) and $A_1$ in \(\Ac'\)
will result in the same 
minimized Mealy automaton $\Ac_4$. The minimized Mealy automaton
$\Bc_3$ is defined 
similarly starting from either $\Bc_1$ or $\Bc'$. Let $\Bc_4$ be the
dual of $\Ac_4$, and let $\Ac_3$ be the dual of $\Bc_3$. 
By construction, the pair $(\Ac_3,\Bc_3)$ (resp. $(\Ac_4,\Bc_4)$) is the one obtained from
$(\Ac,\Bc)$ by minimizing first $\Ac$ (resp. $\Bc$) then $\Bc$
(resp. $\Ac$). But the pair $(\Ac_3,\Bc_3)$ (resp. $(\Ac_4,\Bc_4)$) is
also the one obtained by applying one minimization step starting from
$(\Ac',\Bc')$. Observe that the pair $(\Ac',\Bc')$ has a number of
states strictly smaller than the one of $(\Ac,\Bc)$. By induction
hypothesis, starting from $(\Ac',\Bc')$, 
the $\mz\dz$-reduction does not depend on the
first minimization step, which proves the result.
\end{proof}

\subsection{A sufficient condition for finiteness}

A trivial Mealy automaton  is a Mealy automaton with one state over a
one-letter alphabet. It clearly generates the trivial group.

\begin{theorem}\label{prop-red-finite}
If the $\mz\dz$-reduction of a Mealy automaton (\resp an invertible Mealy automaton)
leads to a trivial Mealy automaton,
then the automaton generates a finite semigroup (\resp a finite group).
\end{theorem}

\begin{proof}
Let $(\Ac,\Bc)$ be a pair of dual Mealy automata and
assume that there exists a sequence of dual Mealy automata
$((\Ac_k,\Bc_k))_{k\in[0,m]}$ such that $(\Ac_0,\Bc_0)=(\Ac,\Bc)$,
$(\Ac_m,\Bc_m)$ is trivial and, for every $k\in[1,m]$,
either $\Ac_k$ is the minimization of $\Ac_{k-1}$
or $\Bc_k$ is the minimization of $\Bc_{k-1}$.

By Proposition~\ref{pr:duale-finitude},
for every $k$, if $\Ac_k$ or $\Bc_k$ generates a finite semigroup, both
automata do.
Obviously, $\Ac_m$ and $\Bc_m$ both generate the trivial group.
We prove that if $\Ac_k$ generates a finite semigroup, so does $\Ac_{k-1}$.
If $\Ac_k$ is the minimization of $\Ac_{k-1}$,
by Lemma~\ref{lem-min}, they both generate the same semigroup.
Otherwise, $\Bc_k$ is the minimization of $\Bc_{k-1}$.
Then $\Bc_k$ generates a finite
semigroup (Prop.~\ref{pr:duale-finitude}),
so does $\Bc_{k-1}$ (Lem.~\ref{lem-min}), and thus
$\Ac_{k-1}$ (Prop.~\ref{pr:duale-finitude}). 
Therefore $\Ac$ generates a finite semigroup.
\end{proof}

Let \(\aut{A}\) be the following automaton:
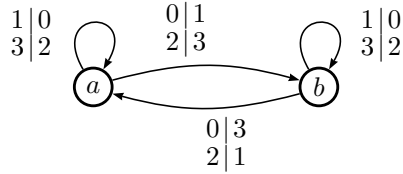
\begin{center}
\MediumPicture\VCDraw{%
\begin{VCPicture}{(-1,-2.5)(6,2.5)}
\State[a]{(0,0)}{A} \State[b]{(5,0)}{B}
\ArcL{A}{B}{\StackTwoLabels{\IOL{0}{1}}{\IOL{2}{3}}}
\ArcL{B}{A}{\StackTwoLabels{\IOL{0}{3}}{\IOL{2}{1}}}
\LoopN[.2]{A}{\StackTwoLabels{\IOL{1}{0}}{\IOL{3}{2}}}
\LoopN[.8]{B}{\StackTwoLabels{\IOL{1}{0}}{\IOL{3}{2}}}
\end{VCPicture}}
\end{center}

Let us compute the $\mz\dz$-reduced automaton of~\(\aut{A}\).

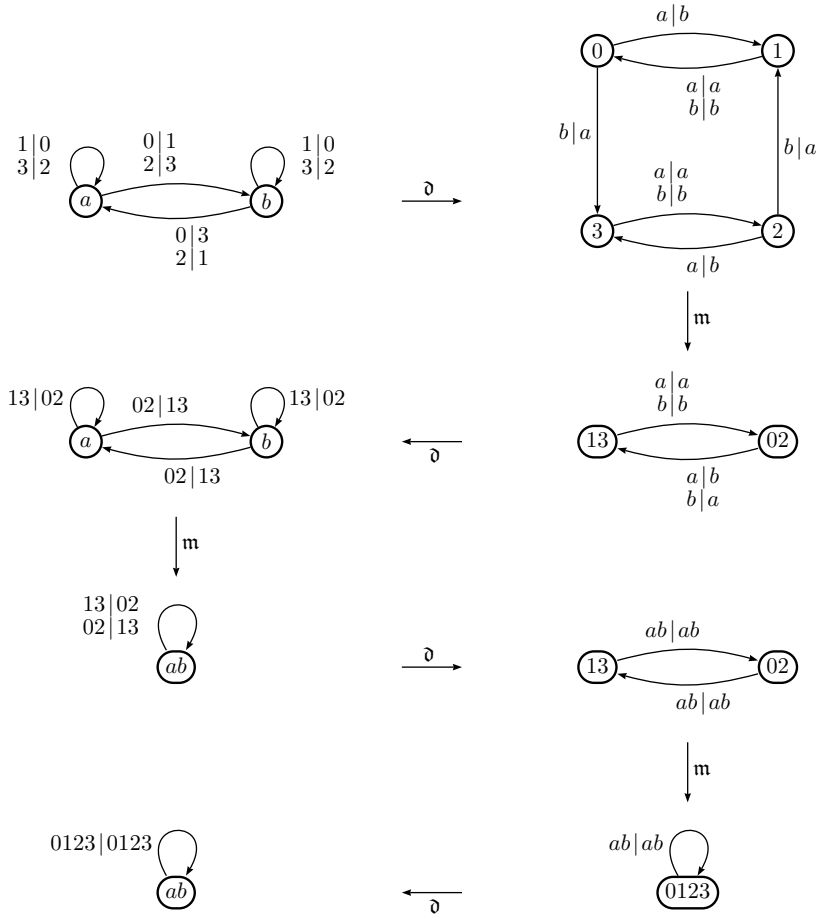
\begin{center}
\SmallPicture
\FixVCGridScale{.8}
\VCDraw{%
\begin{VCPicture}{(-3,-32)(26,-2)}
\State[a]{(0,-8)}{AA} \State[b]{(6,-8)}{BB}
\ArcL{AA}{BB}{\StackTwoLabels{\IOL{0}{1}}{\IOL{2}{3}}}
\ArcL{BB}{AA}{\StackTwoLabels{\IOL{0}{3}}{\IOL{2}{1}}}
\LoopN[.2]{AA}{\StackTwoLabels{\IOL{1}{0}}{\IOL{3}{2}}}
\LoopN[.8]{BB}{\StackTwoLabels{\IOL{1}{0}}{\IOL{3}{2}}}
\Point{(10.5,-8)}{E2} \Point{(12.5,-8)}{F2}
\EdgeL{E2}{F2}{{\mathfrak d}}
%
%
\State[0]{(17,-3)}{A0} \State[1]{(23,-3)}{A1}
\State[3]{(17,-9)}{A3} \State[2]{(23,-9)}{A2}
\ArcL{A1}{A0}{\StackTwoLabels{\IOL{a}{a}}{\IOL{b}{b}}}
\ArcL{A0}{A1}{\IOL{a}{b}}
\ArcL{A3}{A2}{\StackTwoLabels{\IOL{a}{a}}{\IOL{b}{b}}}
\ArcL{A2}{A3}{\IOL{a}{b}}
\EdgeR{A0}{A3}{\IOL{b}{a}}
\EdgeR{A2}{A1}{\IOL{b}{a}}
\Point{(20,-11)}{E3} \Point{(20,-13)}{F3}
\EdgeL{E3}{F3}{{\mathfrak m}}
%
%
\StateVar[13]{(17,-16)}{A13} \StateVar[02]{(23,-16)}{A02}
\ArcL{A13}{A02}{\StackTwoLabels{\IOL{a}{a}}{\IOL{b}{b}}}
\ArcL{A02}{A13}{\StackTwoLabels{\IOL{a}{b}}{\IOL{b}{a}}}
\Point{(12.5,-16)}{E4} \Point{(10.5,-16)}{F4}
\EdgeL{E4}{F4}{{\mathfrak d}}
%
%
\State[a]{(0,-16)}{AAA} \State[b]{(6,-16)}{BBB}
\ArcL{AAA}{BBB}{\IOL{02}{13}}
\ArcL{BBB}{AAA}{\IOL{02}{13}}
\LoopN[.2]{AAA}{\IOL{13}{02}}
\LoopN[.8]{BBB}{\IOL{13}{02}}
\Point{(3,-18.5)}{E5} \Point{(3,-20.5)}{F5}
\EdgeL{E5}{F5}{{\mathfrak m}}
%
%
\StateVar[ab]{(3,-23.5)}{X}
\LoopN[.2]{X}{\StackTwoLabels{\IOL{13}{02}}{\IOL{02}{13}}}
\Point{(10.5,-23.5)}{E6} \Point{(12.5,-23.5)}{F6}
\EdgeL{E6}{F6}{{\mathfrak d}}
%
%
\StateVar[13]{(17,-23.5)}{AAAA} \StateVar[02]{(23,-23.5)}{BBBB}
\ArcL{AAAA}{BBBB}{\IOL{ab}{ab}}
\ArcL{BBBB}{AAAA}{\IOL{ab}{ab}}
\Point{(20,-26)}{E7} \Point{(20,-28)}{F7}
\EdgeL{E7}{F7}{{\mathfrak m}}
%
%
\StateVar[0123]{(20,-31)}{Y}
\LoopN[.2]{Y}{\IOL{ab}{ab}}
\Point{(12.5,-31)}{E8} \Point{(10.5,-31)}{F8}
\EdgeL{E8}{F8}{{\mathfrak d}}
%
%
\StateVar[ab]{(3,-31)}{Z}
\LoopN[.2]{Z}{\IOL{0123}{0123}}
\end{VCPicture}}
\end{center}

The group generated by~$\aut{A}$ is finite and can be shown to
be isomorphic to~$G_{16}^{(9)}$, that is, the group of order~16 with
presentation\[\langle~a,b:a^4=b^4=abab=1,ab^3=ba^3~\rangle.\]

\medskip

Now consider the family~$(\aut{M}^{\sharp}_{p,q})$ of bireversible $p$-letter $q$-state Mealy automata:
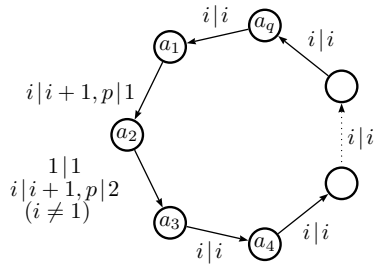
\begin{center}
\SmallPicture
\FixVCGridScale{3}
\VCDraw{\begin{VCPicture}{(-1.5,-1.35)(1.5,1.3)}
\State[a_1]{(-.62,.78)}{A1} \State[a_q]{(.22,.97)}{AQ}
\State[a_2]{(-1,0)}{A2}  \State{(.9,-.43)}{A5}\State{(.9,.43)}{A6}
\State[a_3]{(-.62,-.78)}{A3} \State[a_4]{(.22,-.97)}{A4} 
\EdgeR[.7]{A1}{A2}{\IOL{i}{i+1},\IOL{p}{1}}
\EdgeR[.3]{A2}{A3}{\StackThreeLabels{\IOL{1}{1}}{\IOL{i}{i+1},\IOL{p}{2}}{(i\neq 1)\ \ \ }}
\EdgeR{A3}{A4}{\IOL{i}{i}}
\EdgeR{A4}{A5}{\IOL{i}{i}}
\EdgeR{A6}{AQ}{\IOL{i}{i}}
\EdgeR{AQ}{A1}{\IOL{i}{i}}
\SetEdgeLineStyle{dotted}
\EdgeR{A5}{A6}{\IOL{i}{i}}
\RstEdgeLineStyle
\end{VCPicture}}
\end{center}
\noindent One can check that $(\dz\mz\dz\mz)(\aut{M}^{\sharp}_{p,q})$ is
trivial for any~$p$ and~$q$. Hence by Theorem~\ref{prop-red-finite},
the groups $\pres{\aut{M}^{\sharp}_{p,q}}$ are all finite. 
In fact and independently, the group \(\pres{\aut{M}^{\sharp}_{p,q}}\)
can be identified with \(\perm_q^p\).
For comparison, the packages~$\FR$ and~$\SK$ both fail to decide finiteness of~$\pres{\aut{M}^{\sharp}_{p,q}}$ (except for very small values of~$p,q$).

\subsection{This sufficient condition is not necessary} 

\noindent The following Mealy automaton is $\mz\dz$-reduced, but it
generates a finite semigroup of order 6:
it provides a counterexample to the converse of
Theorem~\ref{prop-red-finite}. 
\begin{center}
\SmallPicture\VCDraw{%
\begin{VCPicture}{(-1.5,-0.5)(6,2)}
\State[a]{(0,0)}{A} \State[b]{(5,0)}{B}
\EdgeL{B}{A}{\IOL{1}{0}}
\LoopN[.2]{A}{\StackTwoLabels{\IOL{0}{1}}{\IOL{1}{1}}}
\LoopN[.8]{B}{\IOL{0}{1}}
\end{VCPicture}}
\end{center}

\medskip\noindent There also exist counterexamples among bireversible Mealy automata. 
Consider the order~8 dihedral group viewed as generated by a reflection~$\s$ and by a product~$\m=\rho\s$ with a rotation:
	\[D_4= \langle~\s,\m:\s^2=\m^2=(\s\m)^4=1~\rangle .\]
It is generated by the bireversible Mealy automaton of
Fig.~\ref{fig-dihedral-automaton}. This ad-hoc automaton is its own dual
and is $\mz\dz$-reduced. 

\begin{figure}[h!]
\SmallPicture\VCDraw{%
\begin{VCPicture}{(2,-2.5)(26,11)}
\StateVar[1]{(8,.5)}{ID}
\StateVar[\m\s\m]{(8,6)}{MSM}
\StateVar[\s\m\s\m]{(5,0)}{SMSM}
\StateVar[\s]{(5,8)}{S}
\StateVar[\m]{(11,8)}{M}
\StateVar[\s\m]{(11,0)}{SM}
\StateVar[\m\s]{(24,8)}{MS}
\StateVar[\s\m\s]{(24,0)}{SMS}
\ChgEdgeLabelScale{.75}
\ChgEdgeLabelSep{.2}
\ForthBackOffset%
\LoopVarN[.3]{ID}{\forall x, \ \IOL{x}{x}}
\LoopVarN[.5]{MSM}{\StackEightLabels
			{\IOL{1}{1}}{\IOL{\s}{\s}}{\IOL{\m}{\s\m\s}}{\IOL{\s\m}{\m\s}}
			{\IOL{\m\s}{\s\m}}{\IOL{\s\m\s}{\m}}{\IOL{\m\s\m}{\m\s\m}}{\IOL{\s\m\s\m}{\s\m\s\m}}}
\EdgeL{S}{SMSM}{\StackFourLabels{\IOL{\m}{\m\s}}{\IOL{\s\m}{\s\m\s}}{\IOL{\m\s}{\m}}{\IOL{\s\m\s}{\s\m}}}
\EdgeL{SMSM}{S}{\StackFourLabels{\IOL{\m}{\s\m}}{\IOL{\s\m}{\m}}{\IOL{\m\s}{\s\m\s}}{\IOL{\s\m\s}{\m\s}}}
\EdgeL[.6]{M}{SM}{\StackTwoLabels		{\IOL{\s\m\s}{\m\s}}	{\IOL{\s\m\s\m}{\s}}}
\EdgeL[.6]{M}{MS}{\StackTwoLabels		{\IOL{\s}{\s\m\s\m}}	{\IOL{\m\s}{\s\m\s}}}
\EdgeL[.3]{SM}{M}{\StackTwoLabels		{\IOL{\s\m\s}{\s\m}}	{\IOL{\s\m\s\m}{\s}}}
\EdgeL[.6]{SM}{SMS}{\StackTwoLabels	{\IOL{\s}{\s\m\s\m}}	{\IOL{\m\s}{\m}}}
\EdgeL[.6]{MS}{M}{\StackTwoLabels		{\IOL{\s}{\s\m\s\m}}		{\IOL{\s\m}{\s\m\s}}}
\EdgeL[.3]{MS}{SMS}{\StackTwoLabels	{\IOL{\m}{\m\s}}		{\IOL{\s\m\s\m}{\s}}}
\EdgeL[.6]{SMS}{SM}{\StackTwoLabels	{\IOL{\s}{\s\m\s\m}}		{\IOL{\s\m}{\m}}}
\EdgeL[.6]{SMS}{MS}{\StackTwoLabels	{\IOL{\m}{\s\m}}		{\IOL{\s\m\s\m}{\s}}}
\LoopVarN[.3]{M}{\StackTwoLabels		{\IOL{1}{1}}		{\IOL{\m}{\m}}}
\LoopVarS[.3]{SM}{\StackTwoLabels		{\IOL{1}{1}}		{\IOL{\m}{\s\m\s}}}
\LoopVarN[.7]{MS}{\StackTwoLabels		{\IOL{1}{1}}		{\IOL{\s\m\s}{\m}}}
\LoopVarS[.7]{SMS}{\StackTwoLabels		{\IOL{1}{1}}		{\IOL{\s\m\s}{\s\m\s}}}

\ChgEdgeLabelSep{-2}
\EdgeL[.7]{M}{SMS}{\StackTwoLabels		{\IOL{\m\s\m}{\m\s\m}}	{\IOL{\s\m}{\s\m}}}
\EdgeL[.75]{SMS}{M}{\StackTwoLabels		{\IOL{\m\s}{\m\s}}		{\IOL{\m\s\m}{\m\s\m}}}
\EdgeL[.7]{SM}{MS}{\StackTwoLabels		{\IOL{\m\s\m}{\m\s\m}}	{\IOL{\s\m}{\m\s}}}
\EdgeL[.7]{MS}{SM}{\StackTwoLabels		{\IOL{\m\s}{\s\m}}		{\IOL{\m\s\m}{\m\s\m}}}
\ChgEdgeLabelSep{-3}
\LoopVarN[.2]{S}{\StackFourLabels{\IOL{\s\m\s\m}{\s\m\s\m}}{\IOL{\m\s\m}{\m\s\m}}{\IOL{\s}{\s}}{\IOL{1}{1}}}
\LoopVarS[.8]{SMSM}{\StackFourLabels{\IOL{1}{1}}{\IOL{\s}{\s}}{\IOL{\m\s\m}{\m\s\m}}{\IOL{\s\m\s\m}{\s\m\s\m}}}
\end{VCPicture}}
\caption{An $\mz\dz$-reduced non-trivial IR-automaton whose group is finite.}\label{fig-dihedral-automaton}
\end{figure}


\section{Helix graphs and finiteness}
\label{s:helix}
In this section, we concentrate on IR-automata and show the pertinence
of helix graphs for the finiteness problem.

\subsection{A necessary condition for finiteness}
To prove the results in this section, 
it is convenient to use a graphical representation in which $A$ and
$\Sigma$ play symmetrical roles. 
Consider $(x,i)\in A \times \Sigma$ with $\delta_i(x)=y$ and
$\rho_x(i)=j$. 
The corresponding transition $x\stackrel{i|j}{\longrightarrow} y$  is
represented by the \emph{cross-transition}:
\[
\croix{x}{y}{i}{j}\:.
\]
The automaton $\aut{A}$ is identified with the set of its
cross-transitions (of cardinality $|A| \times |\Sigma|$).

A path in \(\aut{A}\) (\resp in \(\dz(\aut{A})\)) is represented by
an horizontal (\resp vertical) \emph{cross-diagram} obtained by
concatenating the crosses. We may also consider rectangular
cross-diagrams of dimension $m\times n$, on which one
can read the production functions of $\aut{A}_{m,n}$ and
$\dz(\aut{A}_{m,n})$.
For instance the
cross-diagram:
\begin{minipage}{.4\linewidth}
\[\begin{array}{ccccc}
    & i_1      &     & i_{n} \\
x_1 & \lacroix & \dots & \lacroix & y_1\\
 & \vdots & & \vdots & \\
x_{m} & \lacroix & \dots &  \lacroix & y_{m}\\
    & j_1      &     & j_{n} 
\end{array}\]
\end{minipage}\qquad
\begin{minipage}{.5\linewidth}
corresponds in $\aut{A}_{m,n}$ to 
\[
\rho_{x_1\cdots x_m} (i_1\cdots i_n) = j_1\cdots j_n,\]
\[\delta_{i_1\cdots i_n}(x_1\cdots x_m) = y_1\cdots y_m \:.
\]
\end{minipage}

\noindent Replacing every cross by a square, we get the 
``square-diagrams'' of \cite{square}.

\begin{proposition}\label{prop:helices}
Let \(\aut{A}\) be an IR-automaton. If the helix graph of order
$(1,1)$ of \(\aut{A}\) is a union of cycles, so are all the helix
graphs (of any order) of \(\aut{A}\).
\end{proposition}

\begin{proof}
Observe that a helix graph is a union of cycles if and only if
any node has a predecessor. By assumption, $\cal
H$ is a union of cycles, therefore, any  \((x,
u)\in\alphA\times\alphS\) has a predecessor. Now consider \((\mot{x}, \mot{u})\in\alphA^m\times
\alphS^n\) with \(\mot{x}=x_1\cdots x_m\) and \(\mot{u}=u_1\cdots
u_n\). Let $(\tilde{x}_m,\tilde{u}_n)$ be the predecessor of~$(x_m,u_n)$ in~$\cal H$.
Start with the cross of $(\tilde{x}_m,\tilde{u}_n)$
and $(x_m,u_n)$ (left of~(\ref{eq-cross})), and expand it step-by-step using
the existence of predecessors in~$\cal
H$ (right of~(\ref{eq-cross}) for the first few steps). 
\begin{equation}\label{eq-cross}
\croix{\tilde{x}_{m}}{x_{m}}{\tilde{u}_{n}}{u_{n}}, \qquad \qquad \qquad
\begin{array}{ccccc}
    &  *     &     & * \\
* & \lacroix & * & \lacroix & x_{m-1}\\
 & * & & \tilde{u}_{n} & \\
* & \lacroix & \tilde{x}_{m} &  \lacroix & x_{m}\\
    & u_{n-1}     &     & u_{n} 
\end{array}
\end{equation}
In the end we get a cross-diagram of dimension $m\times n$.
The words on the west and north of the cross-diagram
form a predecessor for~\((\mot{x}, \mot{u})\).
\end{proof}

\begin{theorem}\label{thm:fini_cycles}
Let \(\aut{A}\) be an IR-automaton. If \(\G\) is finite, then the helix
graphs of \(\aut{A}\) are unions of cycles.
\end{theorem}

\begin{proof}
By Proposition~\ref{prop:helices}, it is sufficient to prove the result for
order \((1,1)\). 
Consider  \(x\in \alphA\) and
\(i\in\alphS\).
According to
Proposition~\ref{pr:duale-finitude}, \(\grEng{\dual{{\cal A}}}\) is
finite. Therefore, there exist $m,n>0$ such that \(\rho_x^m=\rho_{x^m}=\id[\grEng{{\cal A}}]\) and
\(\delta_i^n=\delta_{i^n}=\id[\grEng{\dual{{\cal
A}}}]\). It implies that $x^m\stackrel{i^n|i^n}{\longrightarrow} x^m$ is a transition 
in the Mealy automaton of order $(m,n)$. The corresponding
cross-diagram is represented below:
\[\begin{array}{ccccc}
    & i      &     & i \\
x & \lacroix & \dots & \lacroix & x\\
 & \vdots & & \vdots & \\
x & \lacroix & \dots &  \lacroix & x\\
    & i      &     & i 
\end{array}\:.\]
The south-east cross of the diagram provides a predecessor for~$(x,i)$.
\end{proof}

\noindent There exist IR-automata generating infinite groups whose
helix graphs are union of cycles. The smallest examples are Ale\v{s}in
automata (see Table~\ref{tbl-examples}).

Next result follows directly from Lemma~\ref{le-bi} and
Theorem~\ref{thm:fini_cycles}.

\begin{corollary}\label{cor:jir}
Consider an IR-automaton which is not bireversible. Then the group
generated by the automaton is infinite. 
\end{corollary}

\subsection{A necessary and sufficient condition for finiteness}

The condition in next theorem is not effective. Hence, it does not directly lead to a decision procedure of
finiteness.

Recall the construction and notation defined at the end of section
\ref{sse-automgroup}: for an IR-automaton $\aut{A}$ with stateset $A$ and alphabet
$\Sigma$, we denote by $\widetilde{A}$ the extension with stateset $A\sqcup A^{-1}$ and alphabet
$\Sigma\sqcup \Sigma^{-1}$. 

\begin{theorem}\label{th:cycles_bornes}
Consider an IR-automaton $\aut{A}$. The group $\grEng{\aut{A}}$ is
finite if and only if there exists $K$ such that, for all $k,l$, the
helix graphs \(\cH(k,l)\) of  \(\widetilde{\aut{A}}\) are unions of cycles
of lengths bounded by \(K\).
\end{theorem}

\begin{proof}
Assume first that $\G$ is finite: so is \(\pres{\widetilde{\aut{A}}}\) by
Corollary~\ref{cor-gen}. Theorem~\ref{thm:fini_cycles} shows that  
helix graphs of any order are unions of cycles. It remains to prove
that the lengths of these cycles are uniformly bounded. 
By Proposition~\ref{pr:duale-finitude}, the group
$\grEng{\dual{\widetilde{\aut{A}}}}$ is finite as well.  Let \({\cal C}\)
be a cycle in a helix graph of \(\widetilde{\aut{A}}\) and 
let \((\mot{u},\mot{v})\in(\alphA\sqcup\inverse{\alphA})^*
\times(\alphS\sqcup\inverse{\alphS})^*\) be a node of this
cycle. Each node  of  \({\cal C}\) is of the 
form \((h(\mot{u}), g(\mot{v}))\), where \(g\) (\resp \(h\)) is an element of
\(\pres{\widetilde{\aut{A}}}\) (\resp
\(\pres{\dz(\widetilde{\aut{A}})}\)). Since the nodes are pairwise
distinct, the length of the  cycle \({\cal C}\) is at most 
\(\#\pres{\widetilde{\aut{A}}} \times \#\pres{\dz(\widetilde{\aut{A}})}\). 

\medskip

Let us prove the converse and assume that the group \(\G\)
is infinite: so is \(\pres{\widetilde{\aut{A}}}\) by
Corollary~\ref{cor-gen}. 
First we argue  that the orders of the elements of
$\pres{\widetilde{\aut{A}}}$ are unbounded.
Indeed, automata groups are residually finite by construction since they act
faithfully on rooted locally finite trees. Moreover it follows
from Zelmanov's solution of the restricted Burnside problem
\cite{Ze1,Ze2,Vl} that any residually finite group with
bounded torsion is finite. Since  $\pres{\widetilde{\aut{A}}}$ is
infinite, the orders of its 
elements are unbounded. 

There exists either $\mot{x}\in (A\sqcup \inverse{A})^*$ such that the
order of $\rho_{\mot{x}}$ is infinite, or a sequence
\((\mot{x}_n)_{n\in\N}\subseteq(\alphA\sqcup\inverse{\alphA})^*\) such
that the sequence $(k_n)_n$ of orders of the
elements $\rho_{\mot{x}_n}$ converges to infinity. We carry out the
proof in the second case, the first one can be treated similarly. 
Let us concentrate on $\rho_{\mot{x}_n}$, element of order $k_n$ of
$\pres{\widetilde{\aut{A}}}$. For all $1\leq k<k_n$, there exists a word
$\mot{u}_k \in (\Sigma\sqcup\inverse{\Sigma})^*$ such that
$\rho_{\mot{x}_n}^k(\mot{u}_k)=\widetilde{\mot{u}}_k \neq \mot{u}_k$. 

Say that a word $\mot{v}\in (\Sigma\sqcup\inverse{\Sigma})^*$ is {\em
  unitary} if $\delta_{\mot{v}}$ is the identity of 
$\grEng{\dual{\widetilde{\aut{A}}}}$. Since
$\grEng{\dual{\widetilde{\aut{A}}}}$ is a group,
the word $\mot{u}_k$ can be extended into a unitary word
$\mot{u}_k\mot{v}_k$. Set $\mot{w}_n= \mot{u}_1\mot{v}_1\cdots
\mot{u}_{k_n-1}\mot{v}_{k_n-1}$. By construction, we have:
$\rho_{\mot{x}_n}(\mot{w}_n)= \widetilde{\mot{u}}_1\cdots \neq
\mot{w}_n$. Since $\mot{u}_1\mot{v}_1$ is unitary, we also have:
\begin{eqnarray*}
\rho_{\mot{x}_n}^2(\mot{w}_n) & = & \rho_{\mot{x}_n}^2 (\mot{u}_1\mot{v}_1)
\rho_{\mot{x}_n}^2 (\mot{u}_2\mot{v}_2\cdots
\mot{u}_{k_n-1}\mot{v}_{k_n-1}) \\ 
& = & \rho_{\mot{x}_n}^2 (\mot{u}_1\mot{v}_1) \widetilde{\mot{u}}_2 \cdots
\ \neq \ \mot{w}_n\:.
\end{eqnarray*} 
In the same way, we prove that for all $k<k_n$, we have
$\rho_{\mot{x}_n}^k(\mot{w}_n) \neq \mot{w}_n$. 

In the helix graph of \(\widetilde{\aut{A}}\) of order
$(|\mot{x}_n|,|\mot{w}_n|)$,
consider the cycle containing the node \((\mot{x}_n,
\mot{w}_n)\). Since \( \mot{w}_n \) is unitary, the successors of \((\mot{x}_n,
\mot{w}_n)\) on the cycle are: \((\mot{x}_n,
\rho_{\mot{x}_n}(\mot{w}_n))\), \((\mot{x}_n,
\rho_{\mot{x}_n}^2(\mot{w}_n))\), \dots Therefore the cycle is of length~$k_n$. Since $k_n$
converges to infinity, the lengths of the cycles of the helix graphs of~\(\widetilde{\aut{A}}\) are not
uniformly bounded.
\end{proof}

\section{Experimentations}\label{sec-experimentations}

Here, we show how gathering the new criteria with previously known
ones allows to decide the (semi)group finiteness for substantially
more Mealy automata (at least for those with small alphabet and
stateset --- of size up to~3). 

\begin{table}[ht]\label{gag2x2}
\centering
\caption{Results of experimentations on 2-letter 2-state Mealy automata.}
{\begin{tabular}{|c|lE>{\centering }m{\col}|>{\centering }m{\col}|>{\centering }m{\col}|>{\centering }
m{\col}|>{\centering }m{\col}|>{\centering }m{\col}|>{\centering }m{\col}E>{\centering }m{\col}|}
\cline{3-6}
\multicolumn{2}{c|}{}&\multicolumn{4}{c|}{invertible}\tabularnewline
\hline
\multicolumn{2}{|cE}{$2$-letter $2$-state}
					&$\Cijir$	&$\Cji$	&$\Cjir$	&$\Cbir$	&$\Cdijir$	&$\Cdji$	&$\Cnot$	&$\Call$\tabularnewline
\multicolumn{2}{|cE}{Mealy automata}
					&1		&14		&1		&8		&1		&14		&37		&76\tabularnewline
\hline
\multicolumn{4}{c|}{}&\multicolumn{4}{c|}{reversible}\tabularnewline
\cline{5-8}
\multicolumn{10}{c}{\vspace*{-8pt}}\tabularnewline
\hline
\multirow{7}{*}{\rotatebox{90}{\bf\scriptsize \!\!previous criteria}}
&Finitary				&--		&5		&--		&3		&--		&--		&1		&9\tabularnewline
&Thompson-Wielandt\,	&--		&--		&--		&5		&--		&--		&--		&5\tabularnewline
&Level-transitive		&1		&4		&1		&--		&--		&--		&--		&6\tabularnewline
\cline{2-10}
&Sidki				&--	 	&1		&--		&--		&--		&--		&--		&1\tabularnewline
\cline{2-10}
&Limitary cycles		&--	 	&4		&--		&6		&--		&8		&6		&6\tabularnewline
\cline{2-10}
&Cayley${}^\pm$		&1		&1		&1		&--		&--		&1		&2		&6\tabularnewline
&Dual Cayley${}^\pm$	&1	 	&--		&1		&1		&--		&--		&3		&6\tabularnewline
\cline{2-10}
&\quad union			&1	 	&11		&1		&8		&--		&8		&8		&37\tabularnewline
\hline
\hline
\multirow{5}{*}{\rotatebox{90}{\bf\scriptsize \!\!new criteria}}
&$\mz\dz$-trivial		&--		&10		&--		&8		&--		&10		&11		&39\tabularnewline
&Cycles				&1		&--		&1		&--		&--		&--		&--		&2\tabularnewline
&+Sum				&--		&--		&--		&3		&--		&4		&--		&7\tabularnewline
&+Dual				&--		&8		&1		&8		&1		&11		&8		&37\tabularnewline
\cline{2-10}
&\quad union			&1		&10		&1		&8		&1		&14		&13		&48\tabularnewline
\hline
\multicolumn{10}{c}{\vspace*{-8pt}}\tabularnewline
\cline{2-10}
\multicolumn{1}{c|}{}
&\quad total union		&1		&14		&1		&8		&1		&14		&13	
&~\,52\footnotemark[1]\tabularnewline
\cline{2-10}
\end{tabular}}
\end{table}
\vspace*{10pt}

\footnotetext[1]{The table shows that 52 out of 76
  (isomorphism classes of)  2-letter
  2-state Mealy automata can be treated directly using either the old
  or the new criteria. But actually, the finiteness problem is solved
  for the 76 cases. Indeed, 
  a series of papers dealing specifically with 2-letter
  2-state Mealy automata (see~\cite{brs} and references therein) has
  contributed to the actual state of knowledge: 48
  automata  generate finite semigroups, 10 generate semigroups of linear growth,
  17 generate semigroups of exponential
  growth and 1  generates the
  semigroup~$\mathbf{S_{I_2}}$ of intermediate growth (see
  Table~\ref{tbl-examples}).} 

\subsection{Partition}

\noindent For convenience of exposition, we introduce the
decomposition of the whole class of Mealy automata~$\Call$ (up to
isomorphism) into a disjoint union of seven subclasses. By
denoting~$\Ci$ the class of invertible Mealy automata and~$\Cir$ the
class of invertible-reversible Mealy automata, the seven classes are
defined as follows:

\begin{enumerate}
\item[] $\Cbir$ is the class of bireversible Mealy automata,
\item[] $\Cjir$ (standing for $\mathbf J$ust $\Cir$) is the complementary in~$\Cir$ of~$\Cbir$,
\item[] $\Cijir$ consists of the inverses of automata from~$\Cjir$,
\item[] $\Cdijir$ consists of the duals of automata from~$\Cijir$,
\item[] $\Cji$ (standing for $\mathbf J$ust $\Ci$) is the complementary in~$\Ci$ of the union~$\Cir\cup\Cijir$,
\item[] $\Cdji$ consists of the duals of automata from~$\Cji$,
\item[] $\Cnot$ is the complementary (in~$\Call$) of the (disjoint) union of the previous six.
\end{enumerate}

\subsection{Previous criteria}
\subsubsection*{Previously implemented criteria}
The $\GAP$ packages~$\FR$ and~$\SK$
(see~\cite{FR,GAP4,sav}) both overload the functions~\texttt{Order}
and~\texttt{IsFinite} by using several criteria mainly coming from
geometric group theory. More precisely,
we have tested all the corresponding functions:
\texttt{IsFinitaryFRMachine},  \texttt{IsLevelTransitive} and
\texttt{ISFINITE\us THOMPSONWIELANDT\us FR} from $\FR$
and~\texttt{IsFractal} and
\texttt{IsSphericallyTransitive}  from $\SK$. 
While the first two work perfectly, the last three may not stop. From a practical point of view, \texttt{IsSphericallyTransitive} allows to discriminate too few automata.
Now \texttt{IsLevelTransitive} happens to be much slower than
\texttt{IsFractal}, so the latter can be advantageously viewed as a
preliminary criterion of the former. The first half of the~{\bf\small
  previous criteria} part of the following tables expands the
performance of these three criteria coming from geometric group
theory. For~$2$-letter $3$-state and $3$-letter $2$-state (\resp
$3$-letter $3$-state) automata, the execution time
of~\texttt{IsFractal} and~\texttt{IsLevelTransitive} was limited to
100\,000~ms (\resp 200\,000~ms). The resulting data have to be
considered with this arbitrary limitation in mind, together with the
observation that both functions happen to be significantly sensitive
to the representative inside an isomorphism~class. 

\subsubsection*{Sidki's criterion}
Based on Sidki's fundamental work, the solution to the order
problem~\cite{sidkiconjugacy, sidki} for the class of so-called
bounded automorphisms  --- that is, with growth degree at most~0 ---
may provide an infiniteness criterion: in any invertible
automaton~$(A, \Sigma, \delta,\rho)$, a bounded state~$x\in A$ has
infinite order whenever there exists a label~$i|j$ with~$j\not
=i\in\Sigma$ on an edge between~$x$ and some state belonging to the
same strongly connected component. This criterion appears as the second
field of the~{\bf\small previous criteria} part of the tables.  

\subsubsection*{Antonenko's criterion}
An interesting point of view is to investigate those automata~$(A, \Sigma, \delta)$ compelling all the Mealy automata~$(A, \Sigma, \delta,\rho)$ to generate a finite semigroup. A complete characterization of the latter in term of \emph{limitary cycle} given in~\cite{anto} (see also~\cite{russ}) provides a simple effective criterion for finiteness. An automaton is \emph{with limitary cycle} whenever every state~$x\in A$ accessible from some \emph{cyclic} one~$y \in A$ (that is, there exists a nontrivial word~$w\in\Sigma^*$ satisfying~$\delta_w(y)=y$) is \emph{without branch} (that is, $\delta_i(x) = \delta_j(x)$ holds for any~$(i,j)\in\Sigma^2$).
First considered in~\cite{antoberk}, the branchless condition alone is covered by~Proposition~\ref{pr:duale-finitude} and \emph{a fortiori} by Theorem~\ref{prop-red-finite}.
This criterion appears as third field of the~{\bf\small previous criteria} part of the tables. 

\subsubsection*{Maltcev's criterion}
Let $S$ be a finite semigroup. Define the Cayley
machine~$C(S)$ (\resp the dual\footnotemark[2] Cayley
machine~$C^*(S)$) to be the Mealy automaton with stateset $S$,
alphabet $S$, and the following transitions: $\forall x,y \in S$, 
\[
C(S)~: \quad  x
\stackrel{y | xy}{\longrightarrow} xy, \qquad C^*(S) ~: \quad  x
\stackrel{y | yx}{\longrightarrow} xy \:.
\]
\footnotetext[2]{It should be emphasized that
  the current term~\emph{dual} for a Cayley machine is not consistent
  with the widely used term~\emph{dual} for a Mealy automaton.}

According to~\cite{mal} (see also~\cite{min,cain}), for every finite
semigroup~$S$, the semigroup generated by~$C(S)$ (\resp by~$C^*(S)$)
is finite if and only if $S$ is $\mathcal{H}$-trivial (\resp $S$ is
$\mathcal{H}$-trivial and does not contain non-trivial right zero
subsemigroups). 
This can be viewed as an effective finiteness criterion for
those Mealy automata whose isomorphism class intersects the
special class of Cayley machines (\resp dual Cayley machines) and
their possible inverses (which justifies the
symbol~${\tiny\pm}$ in the tables). These two criteria coming from semigroup theory
compose the last quarter of the~{\bf\small previous criteria} part of the tables. 

\subsection{New criteria}
\noindent The first criterion of the~{\bf\small new criteria} part is the $\mz\dz$-triviality from Theorem~\ref{prop-red-finite}.
Next, the criterion Cycles corresponds to Corollary~\ref{cor:jir} which ensures that every automaton from~$\Cijir$ and~$\Cjir$ generates an infinite group.
The last two criteria are ``relative criteria'' --- which vindicates the
symbol~+ --- allowing in good cases to reduce or transpose the
finiteness question to smaller and/or simpler automata. The
criterion~+Sum follows from the easy observation: provided that a Mealy
automaton decomposes into a sum of (smaller) Mealy automata,
it generates an infinite semigroup whenever one sum component does so.  
Finally, the criterion~+Dual follows from~Proposition~\ref{pr:duale-finitude}.

\vbox{
\medbreak\noindent As a simple illustration, let us consider the Mealy
automaton~$\aut{C}$ below on the left. None of the previously known
criteria is suitable to detect the infiniteness
of~$\grEng{\aut{C}}$. Now, the dual~$\dual{\aut{C}}$ happens to be a
sum whose $2$-state component is (isomorphic to) the
dual~$\dual{\aut{B}}$ of the baby~Ale\v{s}in automaton~$\aut{B}$ (see
Table~\ref{tbl-examples}), which turns out to be level-transitive.  

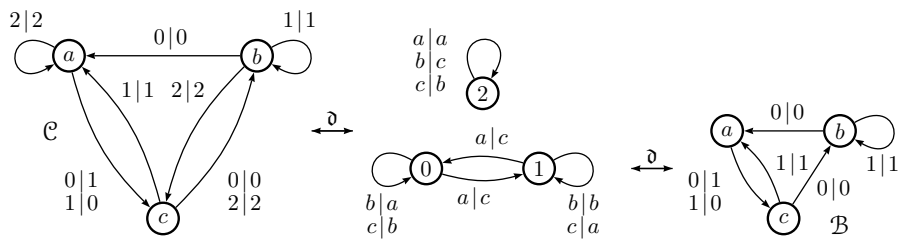
\begin{center}
\SmallPicture\VCDraw{%
\begin{VCPicture}{(-2,-4.9)(22,1.2)}
\State[a]{(0,0)}{A} \State[b]{(5,0)}{B} \State[c]{(2.5,-4.33)}{C}
\ArcR[.7]{A}{C}{\StackTwoLabels{\IOL{0}{1}}{\IOL{1}{0}}}
\EdgeR{B}{A}{\IOL{0}{0}}
\LoopE[.2]{B}{\IOL{1}{1}}
\LoopW[.2]{A}{\IOL{2}{2}}
\ArcR[.3]{C}{B}{\StackTwoLabels{\IOL{0}{0}}{\IOL{2}{2}}}
\ArcR[.3]{B}{C}{\IOL{2}{2}}
\ArcR[.7]{C}{A}{\IOL{1}{1}}
	\HideState
	\State[]{(6,-2)}{X} \State[]{(8,-2)}{Y}
	\EdgeL[.5]{X}{Y}{\dz}
	\EdgeL{Y}{X}{}
	\ShowState
\State[0]{(9.5,-3)}{A0} \State[1]{(12.5,-3)}{A1} \State[2]{(11,-1)}{A2}
\LoopN[.15]{A2}{\StackThreeLabels{\IOL{a}{a}}{\IOL{b}{c}}{\IOL{c}{b}}}
\LoopW[.8]{A0}{\StackTwoLabels{\IOL{b}{a}}{\IOL{c}{b}}}
\LoopE[.8]{A1}{\StackTwoLabels{\IOL{b}{b}}{\IOL{c}{a}}}
\ArcR{A0}{A1}{\IOL{a}{c}}
\ArcR{A1}{A0}{\IOL{a}{c}}
	\HideState
	\State[]{(14.5,-3)}{X} \State[]{(16.5,-3)}{Y}
	\EdgeL[.5]{X}{Y}{\dz}
	\EdgeL{Y}{X}{}
	\ShowState
\State[a]{(17.5,-2)}{AA} \State[b]{(20.5,-2)}{BB} \State[c]{(19,-4.33)}{CC}
\ArcR{AA}{CC}{\StackTwoLabels{\IOL{0}{1}}{\IOL{1}{0}}}
\EdgeR{BB}{AA}{\IOL{0}{0}}
\LoopE[.8]{BB}{\IOL{1}{1}}
\EdgeR{CC}{BB}{\IOL{0}{0}}
\ArcR{CC}{AA}{\IOL{1}{1}}
\put(-0.5,-2){\scalebox{2}{\makebox(0,0){$\aut{C}$}}}
\put(20.5,-4.5){\scalebox{2}{\makebox(0,0){$\aut{B}$}}}
\end{VCPicture}}
\end{center}

\noindent In this way, the isomorphism class of~$\aut{B}$ contributes
for~one in the Level-transitive row only, those of~$\dual{\aut{B}}$
and~$\dual{\aut{C}}$ both contribute for~one in the respective +Dual
rows only and finally that of~$\aut{C}$ contributes for~one in the
+Sum row only. 
}

\begin{table}[ht]\label{gag2x3}
\centering
\caption{Results of experimentations on 2-letter 3-state Mealy automata.}
{\begin{tabular}{|c|lE>{\centering }m{\col}|>{\centering }m{\col}|>{\centering }m{\col}|>{\centering }
m{\col}|>{\centering }m{\col}|>{\centering }m{\col}|>{\centering }m{\col}E>{\centering }m{\col}|}
\cline{3-6}
\multicolumn{2}{c|}{}&\multicolumn{4}{c|}{invertible}\tabularnewline
\hline
\multicolumn{2}{|cE}{$2$-letter $3$-state}
					&$\Cijir$	&$\Cji$	&$\Cjir$	&$\Cbir$	&$\Cdijir$	&$\Cdji$	&$\Cnot$	&$\Call$\tabularnewline
\multicolumn{2}{|cE}{Mealy automata}
					&14		&488		&14		&28		&14		&175		&3270		&4003\tabularnewline
\hline
\multicolumn{4}{c|}{}&\multicolumn{4}{c|}{reversible}\tabularnewline
\cline{5-8}
\multicolumn{10}{c}{\vspace*{-8pt}}\tabularnewline
\hline
\multirow{4}{*}{\rotatebox{90}{\bf\scriptsize \!\!\!\!\!\!prev. crit.}}
&Finitary				&--		&91			&--		&8		&--		&--		&50		&149\tabularnewline
&Thompson-Wielandt\,	&--		&--			&--		&18		&--		&--		&--		&18\tabularnewline
&Level-transitive
					&14		&263		&14		&2		&--		&--		&--		&293\tabularnewline
\cline{2-10}
&Sidki				&--		&35		&--		&--		&--		&--		&--		&35\tabularnewline
\cline{2-10}
&Limitary cycles		&--		&50			&--		&14		&--		&37		&218	&319\tabularnewline
\cline{2-10}
&\quad union			&14		&385		&14		&28		&--		&37		&242	&720\tabularnewline
\cline{2-10}
\hline
\hline
\multirow{5}{*}{\rotatebox{90}{\bf\scriptsize \!\!new criteria}}
&$\mz\dz$-trivial		&--		&194		&--		&26		&--		&55		&386	&661\tabularnewline
&Cycles				&14		&--			&14		&--		&--		&--		&--		&28\tabularnewline
&+Sum				&2		&28			&2		&14		&2		&59		&99		&206\tabularnewline
&+Dual				&--		&132		&14		&21		&14		& 104	&118	&403\tabularnewline
\cline{2-10}
&\quad union			&14		&202		&14		&27		&14		&159	&427	&857\tabularnewline
\hline
\multicolumn{10}{c}{\vspace*{-8pt}}\tabularnewline
\cline{2-10}
\multicolumn{1}{c|}{}
&\quad total union		&14		&466		&14		&28		&14		&159	&519		&1214\tabularnewline
\cline{2-10}
\end{tabular}}
\end{table}
\vspace*{15pt}

\begin{table}[ht]\label{gag3x2}
\centering
\caption{Results of experimentations on 3-letter 2-state Mealy automata.}
{\begin{tabular}{|c|lE>{\centering }m{\col}|>{\centering }m{\col}|>{\centering }m{\col}|>{\centering }
m{\col}|>{\centering }m{\col}|>{\centering }m{\col}|>{\centering }m{\col}E>{\centering }m{\col}|}
\cline{3-6}
\multicolumn{2}{c|}{}&\multicolumn{4}{c|}{invertible}\tabularnewline
\hline
\multicolumn{2}{|cE}{$3$-letter $2$-state}
					&$\Cijir$	&$\Cji$	&$\Cjir$	&$\Cbir$	&$\Cdijir$	&$\Cdji$	&$\Cnot$	&$\Call$\tabularnewline
\multicolumn{2}{|cE}{Mealy automata}
					&14		&175		&14		&28		&14		&488		&3270		&4003\tabularnewline
\hline
\multicolumn{4}{c|}{}&\multicolumn{4}{c|}{reversible}\tabularnewline
\cline{5-8}
\multicolumn{10}{c}{\vspace*{-8pt}}\tabularnewline
\hline
\multirow{4}{*}{\rotatebox{90}{\bf\scriptsize \!\!\!\!\!\!prev. crit.}}
&Finitary				&--		&11		&--		&4		&--		&--		&4		&19\tabularnewline
&Thompson-Wielandt\,	&--		&--		&--		&13		&--		&--		&--		&13\tabularnewline
&Level-transitive		&11		&84		&12		&--		&--		&--		&--		&107\tabularnewline
\cline{2-10}
&Sidki				&--		&2	&--		&--		&--		&--		&--		&2\tabularnewline
\cline{2-10}
&Limitary cycles		&--		&11		&--		&16		&--		&132	&118	&277\tabularnewline
\cline{2-10}
&\quad partial union		&11		&104	&12		&21		&--		&132	&118	&398\tabularnewline
\hline
\hline
\multirow{5}{*}{\rotatebox{90}{\bf\scriptsize \!\!new criteria}}
&$\mz\dz$-trivial		&--		&55		&--		&26		&--		&194	&386	&661\tabularnewline
&Cycles				&14		&--		&14		&--		&--		&--		&--		&28\tabularnewline
&+Sum				&--		&--		&--		&8		&--		&66		&--		&74\tabularnewline
&+Dual				&2		&69		&14		&28		&14		&395	&313	&835\tabularnewline
\cline{2-10}
&\quad partial union		&14		&75		&14		&28		&14		&466	&519	&1130\tabularnewline
\hline
\multicolumn{10}{c}{\vspace*{-8pt}}\tabularnewline
\cline{2-10}
\multicolumn{1}{c|}{}
&\quad total union		&14		&159	&14		&28		&14		&466	&519	&1214\tabularnewline
\cline{2-10}
\end{tabular}}
\end{table}
\vspace*{15pt}

\begin{table}[ht]\label{gag3x3}
\centering
\caption{Results of experimentations on 3-letter 3-state invertible or reversible automata.}
{\begin{tabular}{|c|lE>{\centering }m{\colcol}|>{\centering }m{\colcol}|>{\centering }m{\colcol}|>{\centering }
m{\colcol}|>{\centering }m{\colcol}|>{\centering }m{\colcol}E>{\centering }m{\colcolcol}|}
\cline{3-6}
\multicolumn{2}{c|}{}&\multicolumn{4}{c|}{invertible}\tabularnewline
\hline
\multicolumn{2}{|cE}{$3$-letter $3$-state}
					&$\Cijir$		&$\Cji$		&$\Cjir$		&$\Cbir$		&$\Cdijir$		&$\Cdji$		&$\Call\setminus\Cnot$\tabularnewline
\multicolumn{2}{|cE}{Mealy automata}
					&1073		&116502		&1073		&335		&1073		&116502		&236558\tabularnewline
\hline
\multicolumn{4}{c|}{}&\multicolumn{4}{c|}{reversible}\tabularnewline
\cline{5-8}
\multicolumn{9}{c}{\vspace*{-8pt}}\tabularnewline
\hline
\multirow{7}{*}{\rotatebox{90}{\bf\scriptsize \!\!previous criteria}}
&Finitary				&--			&898		&--			&17			&--			&--			&915\tabularnewline
&Thompson-Wielandt\,	&--			&--			&--			&164		&--			&--			&164\tabularnewline
&Level-transitive
					&996		&71748		&612		&12			&--			&--			&73368\tabularnewline
\cline{2-9}	
&Sidki				&--			&614		&--			&--			&--			&--			&614\tabularnewline
\cline{2-9}	
&Limitary cycles		&--			&627		&--			&68			&--			&3415		&4110\tabularnewline
\cline{2-9}	
&Cayley${}^\pm$		&1			&1			&1			&--			&--			&1			&4\tabularnewline
&Dual Cayley${}^\pm$	&1	 		&--			&1			&1			&--			&--			&3\tabularnewline
\cline{2-9}	
&\quad union			&996 		&73494	&612		&204		&--			&3415		&78721\tabularnewline
\hline
\hline
\multirow{5}{*}{\rotatebox{90}{\bf\scriptsize \!\!new criteria}}
&$\mz\dz$-trivial		&--			&5928		&--			&187		&--			&5928		&12043\tabularnewline
&Cycles				&1073		&--			&1073		&--			&--			&--			&2146\tabularnewline
&+Sum				&76			&736		&76			&109		&76			&9985		&11058\tabularnewline
&+Dual				&76			&11077		&1073		&228		&1073		&73725	&87252\tabularnewline

\cline{2-9}	
&\quad union			&1073		&12811		&1073		&293		&1073		&84601	&100924\tabularnewline
																				\hline
\multicolumn{9}{c}{\vspace*{-8pt}}\tabularnewline
\cline{2-9}	
\multicolumn{1}{c|}{}
&\quad total union		&1073		&84601	&1073		&316		&1073		&84601	&172737\tabularnewline
																				\cline{2-9}	
\end{tabular}}
\end{table}
\vspace*{10pt}

\section{Conclusion}
In this paper, 
we have emphasized the interest of the duality of Mealy automata
for the finiteness problem.
Our new approaches enable to treat a much larger number of Mealy
automata as before, see Section~\ref{sec-experimentations}. We also completely
settle the case of non-bireversible IR-automata (they generate
infinite groups). On the downside, 
the decidability of the finiteness problem remains open. However, we
believe that the characterization in Theorem~\ref{th:cycles_bornes}
could lead to a decision procedure for bireversible automata. Indeed,
experimentations show that 
the cycle-lengths stay almost constant for known finite
groups and increase extremely fast for known infinite groups.

\newpage
\bibliographystyle{plain}
\bibliography{./finiteness}

\end{document}